\documentclass[journal,12pt,onecolumn,draftcls]{IEEEtran}
\IEEEoverridecommandlockouts

\usepackage{amsmath, amsthm, amssymb, amsfonts, mathtools, xargs, tensor, units, cite, stmaryrd, mathrsfs, algpseudocode, comment, graphicx}
\usepackage{multibib}
\usepackage[svgnames]{xcolor}
\usepackage[unicode=true, bookmarks=true, bookmarksnumbered=false, bookmarksopen=false, breaklinks=false, pdfborder={0 0 1}, backref=false, colorlinks=false, hidelinks]{hyperref}
\makeatletter
\let\NAT@parse\undefined
\makeatother

\newtheorem{manuallemmainner}{Lemma}
\newenvironment{manuallemma}[1]{%
  \renewcommand\themanuallemmainner{#1}%
  \manuallemmainner
}{\endmanuallemmainner}
\newtheorem{manualtheoreminner}{Theorem}
\newenvironment{manualtheorem}[1]{%
  \renewcommand\themanualtheoreminner{#1}%
  \manualtheoreminner
}{\endmanualtheoreminner}
\theoremstyle{definition}

\theoremstyle{definition}
\newtheorem{assumption}{Assumption}
\newtheorem{lemma}{Lemma}
\newtheorem{theorem}{Theorem}

\theoremstyle{remark}
\newtheorem{remark}{Remark}

\allowdisplaybreaks
	\title{\Large Safe Model-Based Reinforcement Learning for Systems with Parametric Uncertainties}
\author{\normalsize S M Nahid Mahmud$^{1}$ \and Scott A Nivison$^{2}$ \and Zachary I. Bell$^{2}$ \and Rushikesh Kamalapurkar$^{1}$% <-this % stops a space
\thanks{*This research was supported, in part, by the Air Force Research Laboratories under award number FA8651-19-2-0009. Any opinions, findings, or recommendations in this article are those of the author(s), and do not necessarily reflect the views of the sponsoring agencies.}% 
\thanks{$^{1}$School of Mechanical and Aerospace Engineering, Oklahoma State University, email: {\tt\footnotesize \{nahid.mahmud, abudia@okstate.edu, rushikesh.kamalapurkar\} @okstate.edu}.}%
\thanks{$^{2}$ Air Force Research Laboratories, Florida, USA, email: {
\tt\footnotesize \{scott.nivison, zachary.bell.10\}
 @us.af.mil.}}}
\begin{document}
\maketitle

\begin{abstract}
    Reinforcement learning has been established over the past decade as an effective tool to find optimal control policies for dynamical systems, with recent focus on approaches that guarantee safety during the learning and/or execution phases. In general, safety guarantees are critical in reinforcement learning when the system is safety-critical and/or task restarts are not practically feasible. In optimal control theory, safety requirements are often expressed in terms of state and/or control constraints. In recent years, reinforcement learning approaches that rely on persistent excitation have been combined with a barrier transformation to learn the optimal control policies under state constraints. To soften the excitation requirements, model-based reinforcement learning methods that rely on exact model knowledge have also been integrated with the barrier transformation framework. The objective of this paper is to develop safe reinforcement learning method for deterministic nonlinear systems, with parametric uncertainties in the model, to learn approximate constrained optimal policies without relying on stringent excitation conditions. To that end, a model-based reinforcement learning technique that utilizes a novel filtered concurrent learning method, along with a barrier transformation, is developed in this paper to realize simultaneous learning of unknown model parameters and approximate optimal state-constrained control policies for safety-critical systems. 
    % \tiny \keyFont{\section{Keywords:} Safe Learning, Barrier Transformation, Model-based Reinforcement Learning, Control theory, Parameter Estimation, Nonlinear Systems}
\end{abstract}

\section{Introduction}

Due to advantages such as repeatability, accuracy, and lack of physical fatigue, autonomous systems have been increasingly utilized to perform tasks that are dull, dirty, or dangerous. Autonomy in safety-critical applications such as autonomous driving and unmanned flight relies on the ability to synthesize safe controllers. To improve robustness to parametric uncertainties and changing objectives and models, autonomous systems also need the ability to simultaneously synthesize and execute control policies online and in real time. This paper concerns reinforcement learning (RL), which has been established as an effective tool for safe policy synthesis for both known and uncertain dynamical systems with finite state and action spaces (see, e.g., \cite{SCC.Sutton.Barto1998,SCC.Doya2000}).

RL typically requires a large number of iterations due to sample inefficiency (see, e.g., \cite{SCC.Sutton.Barto1998}). Sample efficiency in RL can be improved using model-based reinforcement learning (MBRL); however, MBRL methods are prone to failure due to inaccurate models (see, e.g., \cite{SCC.Kamalapurkar.Rosenfeld.ea2016, SCC.Kamalapurkar.Walters.ea2016, SCC.Kamalapurkar.Walters.ea2018}). Online MBRL methods that handle modeling uncertainties are motivated by complex tasks that require systems to operate in dynamic environments with changing objectives and system models, where accurate models of the system and environment are generally not available in due to sparsity of data. In the past, MBRL techniques under the umbrella of approximate dynamic programming (ADP) have been successfully utilized to solve reinforcement learning problems online with model uncertainty (see, e.g., \cite{SCC.Modares.Lewis.ea2013,SCC.Kiumarsi.Lewis.ea2014,SCC.Qin.Zhang.ea2014}). ADP utilizes parametric methods such as neural networks (NNs) to approximate the value function, and the system model online. By obtaining an approximation of both the value function and the system model, a stable closed loop adaptive control policy can be developed (see, e.g., \cite{SCC.Vamvoudakis.Vrabie.ea2009, SCC.Lewis.Vrabie2009,SCC.Bertsekas.ea2011,SCC.Bhasin.Kamalapurkar.ea2012,SCC.Liu.Wei2014}). 

Real-world optimal control applications typically include constraints on states and/or inputs that are critical for safety (see, e.g., \cite{SCC.He.Li.ea2017}). ADP was successfully extended to address input constrained control problems in \cite{SCC.Modares.Lewis.ea2013} and \cite{SCC.Vamvoudakis.Miranda.ea2016}. The state-constrained ADP problem was studied in the context of obstacle avoidance in \cite{SCC.Walters.Kamalapurkar.ea2015} and \cite{SCC.Deptula.ea2020}, where an additional term that penalizes proximity to obstacles was added to the cost function. Since the added proximity penalty in \cite{SCC.Walters.Kamalapurkar.ea2015} was finite, the ADP feedback could not guarantee obstacle avoidance, and an auxiliary controller was needed. In \cite{SCC.Deptula.ea2020}, a barrier-like function was used to ensure unbounded growth of the proximity penalty near the obstacle boundary. While this approach results in avoidance guarantees, it relies on the relatively strong assumptions that the value function is continuously differentiable over a compact set that contains the obstacles in spite of penalty-induced discontinuities in the cost function.

Control Barrier Function (CBF) is another approach to guarantee safety in safety-critical systems (see e.g., \cite{SCC.Ames.Xu.ea2017}), with recent applications to the safe reinforcement learning problems (see e.g., \cite{SCC.Choi.Castaneda.ea2020,SCC.Cohen.ea2020,SCC.Marvi.Kiumarsi.ea2020}). \cite{SCC.Choi.Castaneda.ea2020} have addressed the issue of model uncertainty in safety-critical control with an RL-based data-driven approach. A drawback of this approach is that it requires a nominal controller that keeps the system stable during the learning phase, which may not be always possible to design. In \cite{SCC.Marvi.Kiumarsi.ea2020}, the authors proposes a safe off-policy RL scheme which trades-off between safety and performance. In \cite{SCC.Cohen.ea2020} the authors proposes a safe RL scheme in which the proximity penalty approach from \cite{SCC.Deptula.ea2020} is cast into the framework of CBFs. While the control barrier function results in safety guarantees, the existence of a smooth value function, in spite of a nonsmooth cost function, needs to be assumed. Furthermore, to facilitate parametric approximation of the value function, the existence of a forward invariant compact set in the interior of the safe set needs to be established. Since the invariant set needs to be in the interior of the safe set, the penalty becomes superfluous, and safety can be achieved through conventional Lyapunov methods.

This paper is inspired by a safe reinforcement learning technique, recently developed in \cite{SCC.Yang.Vamvoudakis.ea2019}, based on the idea of transforming a state and input constrained nonlinear optimal control problem into an unconstrained one with a type of saturation function, introduced in \cite{SCC.Graichen.Petit.ea2009, SCC.Bechlioulis.Rovithakis.ea2009}. In \cite{SCC.Yang.Vamvoudakis.ea2019}, the state constrained optimal control problem is transformed using a barrier transformation (BT), into an equivalent, unconstrained optimal control problem. Later, a learning technique is used to synthesize the feedback control policy for this unconstrained optimal control problem. The controller for the original system is then derived from the unconstrained approximate optimal policy by inverting the barrier transformation. In \cite{SCC.Greene.Deptula.ea2020}, the restrictive persistence of excitation requirement in \cite{SCC.Yang.Vamvoudakis.ea2019} is softened using model-based reinforcement learning (MBRL), where exact knowledge of the system dynamics is utilized in the barrier transformation.

One of the primary contributions of this paper is a detailed analysis of the connection between the transformed dynamics and the original dynamics, which is missing from results such as \cite{SCC.Yang.Vamvoudakis.ea2019}, \cite{SCC.Greene.Deptula.ea2020}, and \cite{SCC.Yang.Ding.ea2020}. While the stability of the transformed dynamics under the designed controllers is established in results such as \cite{SCC.Yang.Vamvoudakis.ea2019}, \cite{SCC.Greene.Deptula.ea2020}, and \cite{SCC.Yang.Ding.ea2020}, the implications of the behavior of the transformed system on the original system are not examined. In this paper, it is shown that the trajectories of the original system are related to the trajectories of the transformed system via the barrier transformation as long as the trajectories of the transformed system remain bounded.

While the transformation in \cite{SCC.Yang.Vamvoudakis.ea2019} and \cite{SCC.Greene.Deptula.ea2020} results in verifiable safe controllers, it requires exact knowledge of the system model, which is often difficult to obtain. Another primary contribution of this paper is the development of a novel filtered concurrent learning technique for online model learning and its integration with the barrier transformation method to yield a novel MBRL solution to the online state-constrained optimal control problem under parametric uncertainty. The developed MBRL method learns an approximate optimal control policy in the presence of parametric uncertainties for safety critical systems while maintaining stability and safety during the learning phase. The inclusion of filtered concurrent learning makes the controller robust to modeling errors and guarantees local stability under a \emph{finite} (as opposed to \emph{persistent}) excitation condition. 

In the following, the problem is formulated in Section \ref{problem formulation} and the BT is described and analyzed in Section \ref{sec:BT}. A novel parameter estimation technique is detailed in Section \ref{para} and a model-based reinforcement learning technique for synthesizing feedback control policy in the transformed coordinates is developed in Section \ref{Model-Based Reinforcement Learning}. In Section \ref{Stability Analysis}, a Lypaunov-based analysis is utilized to establish practical stability of the closed-loop system resulting from the developed MBRL technique in the transformed coordinates, which guarantees that the safety requirements are satisfied in the original coordinates. Simulation results in Section \ref{Simulation} demonstrate the performance of the developed method and analyze its sensitivity to various design parameters, followed by a comparison of the performance of the developed MBRL approach to an offline pseudospectral optimal control method. Strengths and limitations of the developed method are discussed in Section \ref{Conclusion}, along with possible extensions.

\section{Problem Formulation} \label{problem formulation} 
\subsection{Control objective}\label{control object} 
Consider a continuous-time affine nonlinear dynamical system 
\begin{equation}
        \dot x = f(x)\theta+g(x)u, \label{eq:Dynamics}
\end{equation}
where $x= [x_{1};\hdots;x_{n}] \in \mathbb{R}^{n}$ is the system state, $\theta \in \mathbb{R}^{p}$ are the unknown parameters, $u \in \mathbb{R}^{q}$ is the control input, and the functions $f:\mathbb{R}^{n} \rightarrow \mathbb{R}^{n \times p}$ and $g: \mathbb{R}^{n} \rightarrow \mathbb{R}^{n \times q}$ are known, locally Lipschitz functions.In the following, $[a;b]$ denotes the vector $[a\,\,b]^T$ and $(v)_i$ denotes the $i$th component of the vector $v$. \\ The objective is to design a controller $u$ for the system in \eqref{eq:Dynamics} such that starting from a given feasible initial condition $x^{0}$, the trajectories $x(\cdot)$ decay to the origin
and satisfy $ x_i(t) \in (a_i,A_i),\forall t\geq 0$, where $ i = 1,2,\hdots,n $ and $ a_{i} < 0 < A_{i} $. While MBRL methods such as those detailed in \cite{SCC.Kamalapurkar.Walters.ea2018} guarantee stability of the closed-loop with state constraints are typically difficult to establish without extensive trial and error. In the following, a BT is used to guarantee state constraints. 

\section{Barrier Transformation}\label{sec:BT}
\subsection{Design}
Let the function $b : \mathbb{R} \rightarrow \mathbb{R}$, referred to as the barrier function (BF), be defined as  
\begin{equation}\label{BF1}
    b_{(a_{i},A_{i})}(y) \coloneqq \log \frac{A_{i}(a_{i}- y)}{a_{i}(A_{i}-y)}, \quad \forall i = 1,2, \hdots, n.
\end{equation}
Define $b_{(a,A)}: \mathbb{R}^{n} \rightarrow \mathbb{R}^{n}$ as  $b_{(a,A)} (x) \coloneqq [b_{(a_{1},A_{1})}((x)_{1}); \hdots ; b_{(a_{n},A_{n})}((x)_{n})]$ with $a = [a_{1}; \hdots ;a_{n}]$ and $A = [A_{1}; \hdots ;A_{n}]$. Moreover, the inverse of \eqref{BF1} on the interval $(a_{i},A_{i})$, is given by
\begin{equation}\label{IBF}
b^{-1}_{(a_{i},A_{i})}(y) = a_{i}A_{i}\frac{e^{y} - 1}{a_{i}e^{y} - A_i}. %\forall i = 1, ... , n-1.
\end{equation}
Define $b^{-1}_{(a,A)}: \mathbb{R}^{n} \rightarrow \mathbb{R}^{n}$ as  $b^{-1}_{(a,A)} (s) \coloneqq [b^{-1}_{(a_{1},A_{1})}((s)_{1}); \hdots ; b^{-1}_{(a_{n},A_{n})}((s)_{n})]$. Taking the derivative of \eqref{IBF} with respect to $y$ yields
\begin{equation} \label{BFT}
    \frac{\mathrm{d}b^{-1}_{(a_{i},A_{i})}(y)}{\mathrm{d}y} = \frac{1}{B_i(y)}, \quad \text{where} \quad B_i(y) \coloneqq \frac{a_i^{2} e^{y} - 2a_i A_i + A_{i}^{2} e^{-y}}{A_i a_i^{2} - a_i A_i^{2}}.
\end{equation}
Consider the BF based state transformation  
\begin{equation} 
s_{i} \coloneqq b_{(a_{i},A_{i})}(x_{i}), \quad 
x_{i} = b^{-1}_{(a_{i},A_{i})}(s_{i}),
\end{equation}
where $s\coloneqq \left[s_1,\cdots,s_n\right]$ denotes the transformed state. In the following derivation, whenever clear from the context, the subscripts $a_{i}$ and $A_{i}$ of the BF and its inverse are suppressed for brevity. The time derivative of the transformed state can be computed using the chain rule as $ \dot {s}_{i} = B_i(s_i)\dot x_{i}$
which yields the transformed dynamics
\begin{equation}
    \dot s_{i}  = B_i(s_i)\left(f(x)\theta + g(x)u\right)_{i}. 
\end{equation}

The dynamics of the transformed state can then be expressed as
\begin{equation}\label{eq:BTDynamics}
    \dot s = F(s) + G(s)u,       
\end{equation}
where $F(s) \coloneqq y(s) \theta$, $ \left(y(s)\right)_i := B_i(s_i)\left(f\left(b^{-1}(s)\right)\right)_i
\in \mathbb{R}^{1 \times p}$, and $ \left(G(s)\right)_i \coloneqq B_i(s_i)\left(g\left(b^{-1}(s)\right)\right)_i \in \mathbb{R}^{ 1 \times q}$.

Continuous differentiability of $b^{-1}$ implies that $F$ and $G$ are locally Lipschitz continuous. Furthermore, $f(0) = 0$ along with the fact that $b^{-1}(0) = 0$ implies that $F(0) = 0$. As a result, for all compact sets $\Omega\subset\mathbb{R}^{n}$ containing the origin, $G$ is bounded on $\Omega$ and there exists a positive constant $L_y$ such that $\forall s\in \Omega$, $\|y(s)\| \leq L_{y} \|s\| $. The following section relates the solutions of the original system to the solutions of the transformed system. 

\subsection{Analysis}
In the following lemma, the trajectories of the original system and the transformed system are shown to be related by the barrier transformation provided the trajectories of the transformed system are \emph{complete} (see, e.g., page 33 of \cite{SCC.Sanfelice2021}). The completeness condition is not vacuous, it is not difficult to construct a system where the transformed trajectories escape to infinity in finite time, while the original trajectories are complete. For example, consider the system $ \dot{x} = x + x^2u $ with $x\in\mathbb{R}$ and $ u\in\mathbb{R} $. All nonzero solutions of the corresponding transformed system $ \dot{s} = B_1(s)b^{-1}_{(-0.5,0.5)}(s) + B_1(s) \left(b^{-1}_{(-0.5,0.5)}(s)\right)^2 u $ under the feedback $ \zeta(s,t) = -b^{-1}_{(-0.5,0.5)}(s) $ escape in finite time. However all nonzero solutions of the original system under the feedback $\xi(x,t) = \zeta(b_{(-0.5,0.5)}(x),t) = -x $ converge to either $ -1 $ or $ 1 $.
\begin{lemma}\label{lem:trajectoryRelation}
    If $t \mapsto \Phi\big(t,b(x^{0}),\zeta\big)$ is a complete Carath\'{e}odory solution to \eqref{eq:BTDynamics}, starting from the initial condition $b(x^{0})$, under the feedback policy $(s,t) \mapsto \zeta (s,t)$ and $t \mapsto \Lambda(t,x^{0},\xi)$ is a Carath\'{e}odory solution to \eqref{eq:Dynamics}, starting from the initial condition $x^{0}$, under the feedback policy $(x,t) \mapsto \xi(x,t)$, defined as $\xi(x,t) = \zeta(b(x),t)$, then $\Lambda(\cdot,x^{0}, \xi)$ is complete and $ \Lambda(t,x^{0}, \xi) = b^{-1}\left(\Phi(t,b(x^{0}),\zeta)\right) $ for all $t \in \mathbb{R}_{\geq 0}$.
\end{lemma}
\begin{proof}
See Lemma \ref{lem:trajectoryRelation} in the Appendix.
\end{proof}

Note that the feedback $\xi$ is well-defined at $x$ only if $b(x)$ is well-defined, which is the case whenever $x$ is inside the barrier. As such, the main conclusion of the lemma also implies that $\Lambda(\cdot,x^0,\xi)$ remains inside the barrier. It is thus inferred from Lemma \ref{lem:trajectoryRelation} that if the trajectories of \eqref{eq:BTDynamics} are bounded and decay to a neighborhood of the origin under a feedback policy $(s,t) \mapsto \zeta (s,t)$, then the feedback policy $(x,t) \mapsto \zeta \big(b(x),t \big)$, when applied to the original system in \eqref{eq:Dynamics}, achieves the control objective stated in section \eqref{control object}.

To achieve BT MBRL in the presense of parametric uncertainties, the following section develops a novel parameter estimator. 

\section{Parameter Estimation}\label{para}
The following parameter estimator design is motivated by the subsequent Lyapunov analysis, and is inspired by the finite-time estimator in \cite{SCC.Adetola.Guay2008} and the filtered concurrent learning (FCL) method in \cite{SCC.Roy.Bhasin.ea2016}. Estimates of the unknown parameters, $\hat{\theta} \in \mathbb{R}^{p}$, are generated using the filter
\begin{equation}\label{regressor_Y}
\dot{Y} = \begin{cases}y(s), & \left\Vert Y_{f} \right\Vert \leq \overline{Y_{f}},\\ 0, & \text{otherwise}, \end{cases}\quad Y(0) = 0,   
\end{equation}
\begin{equation}\label{regressor_Y_f}
\dot{Y}_{f} = \begin{cases}Y^{T}Y, & \left\Vert Y_f \right\Vert \leq \overline{Y_f},\\ 0, & \text{otherwise}, \end{cases} \quad Y_{f}(0) = 0,   
\end{equation}
\begin{equation}\label{regressor_G_f}
\dot{G}_{f} = \begin{cases}G(s)u, & \left\Vert Y_f \right\Vert \leq \overline{Y_f},\\ 0,&\text{otherwise},\end{cases}, \quad G_{f}(0) = 0,
\end{equation}
\begin{equation}\label{regressor_X_f}
\dot{X}_{f} = \begin{cases}Y^{T}(s-s^{0}-G_{f}), & \left\Vert Y_f \right\Vert \leq \overline{Y_f},\\ 0,&\text{otherwise},\end{cases} \quad X_{f}(0)=0, 
\end{equation}
where $s^{0} = \left[b\left(x^{0}_{1}\right);\hdots;b\left(x^{0}_{n}\right)\right]$, and the update law 
\begin{equation}\label{theta_update}
\dot{\hat{\theta}} = \beta_{1}Y_{f}^T(X_{f}-Y_{f}{\hat{\theta}}), \quad  \hat{\theta}(0) = \theta^{0},
\end{equation}
where $ \beta_{1}$ is a symmetric positive definite gain matrix and $\overline{Y_f}$ is a tunable upper bound on the filtered regressor $Y_f$. 

Equations \eqref{eq:BTDynamics} - \eqref{theta_update} constitute a nonsmooth system of differential equations
\begin{equation}\label{z_dot}
    \dot{z} = h(z,u) = \begin{cases}h_{1}(z,u), & \left\Vert Y_f \right\Vert \leq \overline{Y_f},\\ h_{2}(z,u),&\text{otherwise},\end{cases} 
\end{equation}
where $z = [s; \ \text{vec}(Y); \ \text{vec}(Y_{f}); \ G_{f}; \ X_{f}; \ \hat{\theta}]$, $h_{1}(z,u) = [F(s)+G(s)u; \ \text{vec}(y(s)); \text{vec}(Y^{T}Y); \\ G(s)u; Y^{T}(s-s^{o}-G_{f}); \ \beta_{1}Y_{f}^{T}(X_{f}-Y_{f}\hat{\theta})]$,  and $h_{2}(z,u) = [F(s)+G(s)u; \ 0; \ 0; \ 0; \ 0; \ \beta_{1}Y_{f}^{T}(X_{f}-Y_{f}\hat{\theta})]$. Since $\|Y_{f}\|$ is non-decreasing in time, it can be shown that \eqref{z_dot} admits Carath\'{e}odory solutions.

\begin{lemma}\label{existence_Caratheodory}
    If $\|Y_{f}\|$ is non-decreasing in time then \eqref{z_dot} admits Carath\'{e}odory solutions.
\end{lemma}

\begin{proof}
see Lemma \ref{existence_Caratheodory} in Appendix.
\end{proof}

Note that \eqref{regressor_Y_f}, expressed in the integral form 
\begin{equation}\label{integral_Y_f}
    {Y}_{f}(t) = \int_{0}^{t_{3}}Y^{T}(\tau)Y(\tau)d\tau, 
\end{equation}
where $t_{3} \coloneqq \underset{t} {\inf}\{t \geq 0 \quad | \quad \|Y_{f}(t)\| \leq \overline{Y_{f}}\}$, along with \eqref{regressor_X_f}, expressed in the integral form 
\begin{equation}\label{integral_X_f}
    {X}_{f}(t) = \int_{0}^{t_{3}}Y^{T}(\tau)\left(s(\tau)-s^{0}-G_{f}(\tau)\right)d\tau,
    \end{equation}
and the fact that $s(\tau) - s^0 - G_f(\tau) = Y(\tau) \theta$, can be used to conclude that ${X}_{f}(t) = Y_{f}(t)\theta $, for all $t\geq 0$. As a result, a measure for the parameter estimation error can be obtained using known quantities as $ Y_{f}\tilde{\theta} =  {X}_{f} - Y_{f}\hat{\theta} $, where $\tilde{\theta} := \theta - \hat{\theta}.$ The dynamics of the parameter estimation error can then be expressed as 
\begin{equation}
    \dot {\tilde{\theta}} = -\beta_{1}Y_{f}^TY_{f}{\tilde{\theta}}. \label{eq:theta_tilde_dot}
\end{equation}
The filter design is thus motivated by the fact that if the matrix $Y_{f}^TY_{f}$ is positive definite, uniformly in $t$, then the Lyapunov function $V_{1}(\tilde{\theta}) =  \frac{1}{2} \tilde{\theta}^{T}\beta_{1}^{-1}\tilde{\theta}$ can be used to establish convergence of the parameter estimation error to the origin. Initially, $Y_{f}^TY_{f}$ is a matrix of zeros. To ensure that there exists some finite time $T$ such that $Y_{f}^T(t)Y_{f}(t)$ is positive definite, uniformly in $t$ for all $t\geq T$, the following \emph{finite} excitation condition is imposed.
\begin{assumption} 
    There exists a time instance $T>0$ such that $ Y_f(T)$ is full rank.\label{ass:finite_excitation}
\end{assumption}
Note that the minimum eigenvalue of $Y_f$ is trivially non-decreasing for $t\geq t_{3}$ since $Y_{f}(t)$ is constant $\forall t \geq t_{3}$. For $t_{4} \leq t_{5} \leq t_{3}$, $Y_f(t_5) = Y_f(t_4) +  \int_{t_4}^{t_5} Y^{T}(\tau)Y(\tau)d\tau$. Since $Y_f(t_4)$ is positive semidefinite, and so is the integral $\int_{t_4}^{t_5} Y^{T}(\tau)Y(\tau)d\tau$, we conclude that $\lambda_{\min}(Y_f(t_5))\geq \lambda_{\min}(Y_f(t_4)) $, As a result, $t\mapsto\lambda_{\min}(Y_{f}(t))$ is non-decreasing. Therefore, if Assumption \ref{ass:finite_excitation} is satisfied at $t=T$, then $Y_f(t)$ is also full rank for all $t\geq T$. Similar to other MBRL methods that rely on system identification  (see e.g., chapter 4 of \cite{SCC.Kamalapurkar.Walters.ea2018}) the following assumption is needed to ensure boundedness of the state trajectories over the interval $ [0,T] $.
\begin{assumption}
    A fallback controller $ \psi:\mathbb{R}^n \times \mathbb{R}_{\geq 0}\to\mathbb{R}^q $ that keeps the trajectories of \eqref{eq:BTDynamics} inside a known bounded set over the interval $[0,T)$, without requiring the knowledge of $ \theta $, is available. \label{ass:known_controller}
\end{assumption}
If a fallback controller that satisfies Assumption \ref{ass:known_controller} is not available, then, under the additional assumption that the trajectories of \eqref{eq:BTDynamics} are exciting over the interval $ [0,T) $, such a controller can be learned, online while maintaining system stability, using model-free reinforcement learning techniques such as \cite{SCC.Bhasin.Kamalapurkar.ea2013a, SCC.Vrabie.Lewis.ea2010} and \cite{SCC.Modares.Lewis.ea2014}.
 
\begin{remark}
    While the analysis of the developed technique dictates that a different stabilizing controller should be used over the time interval $[0,T)$, typically, similar to the examples in Sections \ref{simsec1} and \ref{simsec2}, the transient response of the developed controller provides sufficient excitation so that $T$ is small (in the examples provided in Sections \ref{simsec1} and \ref{simsec2}, $T$ is the order of $10^{-5}$ and $10^{-6}$, respectively), and a different stabilizing controller is not needed in practice.
\end{remark}

\section{Model-Based Reinforcement Learning}\label{Model-Based Reinforcement Learning}
Lemma \ref{lem:trajectoryRelation} implies that if a feedback controller that practically stabilizes the transformed system in \eqref{eq:BTDynamics} is designed, then the same feedback controller, applied to the original system by inverting the BT, also achieves the control objective stated in Section \ref{control object}. In the following, a controller that practically stabilizes \eqref{eq:BTDynamics} is designed as an estimate of the controller that minimizes the infinite horizon cost. 
\begin{figure}
    \centering
	\includegraphics[width=0.75\columnwidth]{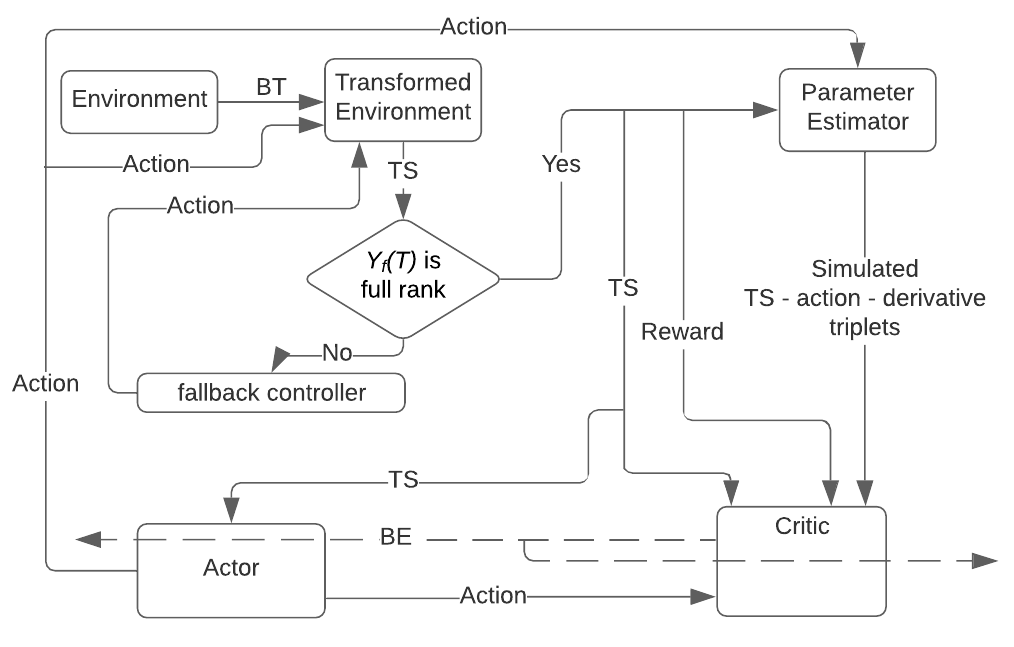}
	\caption{The developed BT MBRL framework. The control system consists of a model-based barrier-actor-critic-estimator architecture. In addition to the transformed state-action measurements, the critic also utilizes states, actions, and the corresponding state derivatives, evaluated at arbitrarily selected points in the state space, to learn the value function. In the figure, BT: Barrier Transformation; TS: Transformed State; BE: Bellman Error.}
	\label{fig:TNNS}
\end{figure}
\begin{equation} \label{cost function}
    J(u(\cdot)) \coloneqq 	\int_{0}^\infty r(\phi(\tau,s^0,u(\cdot)), u(\tau)) d\tau,
\end{equation}
over the set $\mathcal{U}$ of piecewise continuous functions $t\mapsto u(t)$, subject to \eqref{eq:BTDynamics}, where $\phi(\tau, s^0, u (\cdot))$ denotes the trajectory of
\eqref{eq:BTDynamics}, evaluated at time $\tau$, starting from the state $s^0$ and
under the controller $u (\cdot)$, $r(s,u) \coloneqq s^{T}Qs + u^{T}Ru$, and $Q$ $\in$  $\mathbb{R}^{n \times n}$ and $R \in \mathbb{R}^{q \times q}$ are symmetric positive definite (PD) matrices\footnote{  For ease of exposition, a state penalty of the form $s^{T} Q s$ has been considered in this paper. However, the analysis extends in a straightforward manner to general positive definite state penalty functions $s\mapsto Q(s)$. As such, a state penalty function $x\mapsto P(x)$, given in the original coordinates, can easily be transformed into an equivalent state penalty $Q(s) = P(b^{-1}(s))$. Since the barrier function is monotonic and $b(0) = 0$, if $P$ is positive definite, then so is $Q$. Furthermore, for applications with bounded control inputs, a non-quadratic penalty function similar to Eq. 17 of \cite{SCC.Yang.Ding.ea2020} can be incorporated in \eqref{cost function}.}.

Assuming that an optimal controller exists, let the optimal value function, denoted by $V^{*} : \mathbb{R}^{n} \times \mathbb{R}^q  \rightarrow \mathbb{R} $, be defined as
\begin{equation}
    V^{*}(s) := \min_{u(\cdot)\in \mathcal{U}_{[t,\infty})}\int_{t}^\infty r(\phi(\tau,s,u_{[t,\tau)}(\cdot)), u(\cdot)) d\tau,
\end{equation}
where $u_I$ and $\mathcal{U}_I$ are obtained by restricting the domains of $u$ and functions in $\mathcal{U}_I$ to the interval $ I \subseteq \mathbb{R} $, respectively. Assuming that the optimal value function is continuously differentiable, it can be shown to be the unique positive definite solution of the Hamilton-Jacobi-Bellman (HJB) equation 
(see, e.g., \cite{SCC.Greene.Deptula.ea2021})
\begin{equation} \label{HJB} 
    \min_{u\in\mathbb{R}^q} \left(\nabla_{s}V(s)\left(F(s)+G(s)u\right)+ s^{T}Qs + u^{T}Ru\right) = 0,
\end{equation}
where $\nabla_{\left(\cdot \right)} \coloneqq  \frac{\partial}{\partial {\left(\cdot \right)} }$. Furthermore, the optimal controller is given by the feedback policy $u(t) = u^*(\phi(t,s,u_{[0,t)}))$ where $ u^{*}: \mathbb{R}^{n} \rightarrow \mathbb{R}^{q} $ defined as 
\begin{equation}
  u^{*}(s) := -\frac{1}{2}R^{-1}G(s)^{T}(\nabla_sV^{*}(s))^{T}.
\end{equation}
\begin{remark}
    In the developed method, the cost function is selected to be quadratic in the transformed coordinates. However, a physically meaningful cost function is more likely to be available in the original coordinates. If such a cost function is available, it can be transformed from the original coordinates to the barrier coordinates using the inverse barrier function, to yield a cost function that is not quadratic in the state. While the analysis in this paper addresses the quadratic case, it can be extended to address the non-quadratic case with minimal modifications as long as $ s \mapsto r(s,u) $ is positive definite for all $u \in \mathbb{R}^q$.
\end{remark}

\subsection{Value function approximation}
Since computation of analytical solutions of the HJB equation is generally infeasible, especially for systems with uncertainty, parametric approximation methods are used to approximate the value function $V^{*}$ and the optimal policy $u^{*}$. The optimal value function is expressed as
\begin{equation} \label{eq:optimalV}
    V^{*}\left(s\right)=W^{T}\sigma\left(s\right)+\epsilon\left(s\right),
\end{equation}
where $W\in\mathbb{R}^{L}$ is an unknown vector of bounded weights, $\sigma:\mathbb{R}^{n}\rightarrow\mathbb{R}^{L}$ is a vector of continuously differentiable nonlinear activation functions such that $\sigma\left(0\right)=0$ and $\nabla_{s} \sigma \left(0\right)=0$, $L\in\mathbb{N}$ is the number of basis functions, and $\epsilon:\mathbb{R}^{n}\rightarrow\mathbb{R}$ is the reconstruction error.

The basis functions are selected such that the approximation of the functions and their derivatives is uniform over the compact set $\chi \subset\mathbb{R}^{n}$ so that given a positive constant $\overline{\epsilon}\in\mathbb{R}$, there exists $L\in\mathbb{N}$ and known positive constants $\bar{W}$ and $\overline{\sigma}$ such that $\left\Vert W\right\Vert \leq\bar{W}$, $\sup_{s\in\chi}\left \| \epsilon \left(s\right)\right\| \leq\overline{\epsilon}$, $\sup_{s\in\chi}\left\|\nabla_{s}\epsilon\left(s\right)\right\| \leq\overline{\epsilon}$, $\sup_{s\in\chi}\left \| \sigma \left(s\right)\right\| \leq\overline{\sigma}$, and $\sup_{s\in\chi}\left\|\nabla_{s}\sigma\left(s\right)\right\| \leq\overline{\sigma}$ (see, e.g., \cite{SCC.Sauvigny2012}). Using \eqref{HJB}, a representation of the optimal controller using the same basis as the optimal value function is derived as
\begin{equation}\label{eq:optimalu}
    u^{*}\left(s\right)=-\frac{1}{2}R^{-1}G^{T}\left(s\right)\left(\nabla_{s}\sigma^{T}\left(s\right)W+\nabla_{s}\epsilon^{T}\left(s\right)\right).
\end{equation}
Since the ideal weights, $W$, are unknown, an actor-critic approach is used in the following to estimate $W$. To that end, let the NN estimates $\hat{V}:\mathbb{R}^{n}\times\mathbb{R}^{L}\to\mathbb{R}$ and $\hat{u} : \mathbb{R}^{n} \times \mathbb{R}^{L} \to \mathbb{R}^{q}$ be defined as
\begin{gather}
    \hat{V}\left(s,\hat{W}_{c}\right)\coloneqq\hat{W}_{c}^{T}\sigma\left(s\right),\label{V_app}\\
    \hat{u}\left(s,\hat{W}_{a}\right)\coloneqq-\frac{1}{2}R^{-1}G^{T}\left(s\right)\nabla_{s}\sigma^{T}\left(s\right)\hat{W}_{a},\label{u_app}
\end{gather}
where the critic weights, $\hat{W}_{c}\in\mathbb{R}^{L}$ and actor weights, $\hat{W}_{a}\in\mathbb{R}^{L}$ are estimates of the ideal weights, $W$.

\subsection{Bellman Error}
Substituting \eqref{V_app} and \eqref{u_app} into \eqref{HJB} results in a residual term, $\hat{\delta}: \mathbb{R}^{n} \times \mathbb{R}^{L} \times \mathbb{R}^{L} \times \mathbb{R}^{p} \rightarrow \mathbb{R}$, referred to as Bellman Error (BE), defined as
\begin{equation} \label{BE1}                  
    \hat{\delta}(s,\hat{W}_{c},\hat{W}_{a},\hat{\theta}) \coloneqq\nabla_{s}\hat{V}(s,\hat{W}_{c})\left(y(s)\hat{\theta} +G(s)\hat{u}(s,\hat{W}_{a})\right) + \hat{u}(s,\hat{W}_{a})^{T}R\hat{u}(s,\hat{W}_{a})+s^{T}Qs.
\end{equation}
Traditionally, online RL methods require a persistence of excitation (PE) condition to be able learn the approximate control policy (see, e.g., \cite{SCC.Modares.Lewis.ea2013, SCC.Kamalapurkar.Rosenfeld.ea2016, SCC.Kiumarsi.Lewis.ea2014}). Guaranteeing PE a priori and verifying PE online are both typically impossible. However, using virtual excitation facilitated by model-based BE extrapolation, stability and convergence of online RL can established under a PE-like condition that, while impossible to guarantee a priori, can be verified online (by monitoring the minimum eigenvalue of a matrix in the subsequent Assumption \ref{ass:CLBCADPLearnCond} (see, e.g., \cite{SCC.Kamalapurkar.Walters.ea2016}). Using the system model, the BE can be evaluated at any arbitrary point in the state space. Virtual excitation can then be implemented by selecting a set of states $\left\{ s_{k}\mid k=1,\cdots,N\right\} $ and evaluating the BE at this set of states to yield
\begin{equation} \label{BE2} 
    \hat{\delta}_{k}(s_{k},\hat{W}_{c},\hat{W}_{a},\hat{\theta}) \coloneqq \nabla_{s_{k}}\hat{V}(s_{k},\hat{W}_{c})\big(y_{k}\hat{\theta} +G_{k}\hat{u}(s_{k},\hat{W}_{a})\big) + \hat{u}(s_{k},\hat{W}_{a})^{T}R\hat{u}(s_{k},\hat{W}_{a})+s_{k}^{T}Qs_{k},
\end{equation}
where, $\nabla_{s_{k}} \coloneqq  \frac{\partial}{\partial s_{k}}$, $y_{k} \coloneqq y(s_{k})$ and $G_{k} \coloneqq G\left(s_{k}\right)$. Defining the actor and critic weight estimation errors as $\tilde{W}_{c} \coloneqq W -\hat{W}_{c}$ and  $\tilde{W}_{a} \coloneqq W -\hat{W}_{a}$ and substituting the estimates \eqref{eq:optimalV} and \eqref{eq:optimalu} into \eqref{HJB}, and subtracting from \eqref{BE1} yields the analytical BE that can be expressed in terms of the weight estimation errors as
\begin{equation} \label{Analytical BE}
    \hat{\delta}=-\omega^{T}\tilde{W}_{c}+\frac{1}{4}\tilde{W}_{a}^{T}G_{\sigma}\tilde{W}_{a}-W^{T} \nabla_{s} \sigma y\tilde{\theta}+\Delta,
\end{equation}
where $\Delta\coloneqq\frac{1}{2}W^{T}\nabla_{s}\sigma G_{R}\nabla_{s}\epsilon^{T}+\frac{1}{4}G_{\epsilon}- \nabla_{s} \epsilon F$, $G_{R}\coloneqq GR^{-1}G^{T}\in\mathbb{R}^{n\times n}$,  $G_{\epsilon}\coloneqq \nabla_{s} \epsilon  G_{R} \nabla_{s} \epsilon^{T}\in\mathbb{R}$, $G_{\sigma}\coloneqq \nabla_{s} \sigma G R^{-1}G^{T} \nabla_{s} \sigma^{T} \in\mathbb{R}^{L\times L}$, and  $\omega \coloneqq \nabla_{s} \sigma \left(y\hat{\theta} + G\hat{u}(s,\hat{W}_{a})\right)\in\mathbb{R}^{L}$. In \eqref{Analytical BE} and the rest of the manuscript, the dependence of various functions on the state, $s$, is omitted for brevity whenever it is clear from the context. Similarly, \eqref{BE2} implies that
\begin{equation} \label{Approximate BE}
    \hat{\delta}_{k}=-\omega_{k}^{T}\tilde{W}_{c}+\frac{1}{4}\tilde{W}_{a}^{T}G_{\sigma_{k}}\tilde{W}_{a}-W^{T} \nabla_{s_{k}} \sigma_{k} y_{k}\tilde{\theta}+\Delta_{k},
\end{equation}
where, $F_{k} \coloneqq F(s_{k})$, $\epsilon_{k} \coloneqq \epsilon(s_{k})$, $\sigma_{k} \coloneqq \sigma (s_{k})$, $\Delta_{k} \coloneqq \frac{1}{2}W^{T} \nabla_{s_{k}}\sigma_{k} G_{R_{k}} \nabla_{s_{k}} \epsilon_{k}^{T}+\frac{1}{4}G_{\epsilon_{k}} - \nabla_{s_{k}} \epsilon_{k} F_{k}$, $G_{\epsilon_{k}}\coloneqq \nabla_{s_{k}} \epsilon_{k} G_{R_{k}} \nabla_{s_{k}} \epsilon_{k}^{T}$, $\omega_{k} \coloneqq \nabla_{s_{k}} \sigma_{k}\left(y_{k}\hat{\theta}+  G_{k}\hat{u}(s_{k},\hat{W}_{a})\right)\in\mathbb{R}^{L}$, $G_{R_{k}}\coloneqq G_{k}R^{-1}G_{k}^{T}\in\mathbb{R}^{n\times n}$ and $G_{\sigma_{k}} \coloneqq \nabla_{s_{k}} \sigma_{k} G_{k} R^{-1} G_{k}^{T} \nabla_{s_{k}} \sigma_{k}^{T}\in\mathbb{R}^{L\times L}$. Note that $\sup_{s\in\chi}\left\Vert\Delta\right\Vert\leq d \overline{\epsilon}$ and if $s_{k} \in \chi$ then $ \left\Vert\Delta_{k} \right\Vert\leq d \overline{\epsilon}_{k}$, for some constant $d > 0$.
While the extrapolation states $s_k$ are assumed to be constant in this analysis for ease of exposition, the analysis extends in a straightforward manner to time-varying extrapolation states that are confined to a compact neighborhood of the origin.
\medskip
\subsection{Update laws for Actor and Critic weights}
The actor and the critic weights are held at their initial values over the interval $[0,T)$ and starting at $t=T$, using the instantaneous BE $\hat{\delta}$ from \eqref{BE1} and extrapolated BEs $\hat{\delta}_{k}$ from \eqref{BE2}, the weights are updated according to
\begin{align}
    \dot{\hat{W}}_{c} &= -k_{c_{1}} \Gamma \frac{\omega}{\rho} \hat{\delta} -  \frac{k_{c_{2}}}{N} \Gamma \sum_{k=1}^{N} \frac{\omega_{k}}{\rho_{k}} \hat\delta_{k}, \label{W_c}\\
    \dot{\Gamma} &= \beta\Gamma-k_{c_{1}}\Gamma\frac{\omega\omega^{T}}{\rho^{2}}\Gamma - \frac{k_{c_{2}}}{N}\Gamma\sum_{k=1}^{N}\frac{\omega_{k}\omega_{k}^{T}}{\rho_{k}^{2}}\Gamma,\label{gamma}\\
    \dot{\hat{W}}_{a} &= -k_{a_{1}}\left(\hat{W}_{a}-\hat{W}_{c}\right)-k_{a_{2}}\hat{W}_{a} + \frac{k_{c_{1}}G_{\sigma}^{T}\hat{W}_{a}\omega^{T}}{4\rho}\hat{W}_{c} + \sum_{k=1}^{N}\frac{k_{c_{2}}G_{\sigma_{k}}^{T}\hat{W}_{a}\omega_{k}^{T}}{4N\rho_{k}}\hat{W}_{c},\label{W_a}
\end{align}
with $\Gamma\left(t_{0}\right)=\Gamma_{0}$, where $\Gamma:\mathbb{R}_{\geq t_{0}}\to\mathbb{R}^{L\times L}$
is a time-varying least-squares gain matrix, $\rho\left(t\right)\coloneqq1+\gamma_{1}\omega^{T}\left(t\right)\omega\left(t\right)$,
$\rho_{k}\left(t\right)\coloneqq1+\gamma_{1}\omega_{k}^{T}\left(t\right)\omega_{k}\left(t\right)$, $\beta>0\in\mathbb{R}$ is a constant forgetting factor, and $k_{c_{1}},k_{c_{2}},k_{a_{1}},k_{a_{2}}>0\in\mathbb{R}$ are constant adaptation gains. The control commands sent to the system are then computed using the actor weights as 
\begin{equation}\label{u_command}
    u(t)=\begin{cases} \psi(s(t),t), & 0<t<T,\\
    \hat{u}\left(s(t),\hat{W}_{a}(t)\right), & t\geq T,\end{cases}
\end{equation}
where the controller $\psi$ was introduced in Assumption \ref{ass:finite_excitation}. The following verifiable PE-like rank condition is then utilized in the stability analysis. 
\begin{assumption}
    \label{ass:CLBCADPLearnCond}There exists a constant $\underline{c}_{3} > 0$ such that the set of points $\left\{ s_{k}\in\mathbb{R}^{n}\mid k=1,\hdots,N\right\} $ satisfies
    \begin{equation}
    \underline{c}_{3}I_{L} \leq\inf_{t\in\mathbb{R}_{\geq T}}\left(\frac{1}{N}\sum_{k=1}^{N}\frac{\omega_{k}\left(t\right)\omega_{k}^{T}\left(t\right)}{\rho_{k}^{2}\left(t\right)}\right).\label{eq:CLBCPE2}
    \end{equation}
\end{assumption}
Since $\omega_{k}$ is a function of the weight estimates $\hat{\theta}$ and $\hat{W}_{a}$,  Assumption \ref{ass:CLBCADPLearnCond} cannot be guaranteed a priori. However, unlike the PE condition, Assumption \ref{ass:CLBCADPLearnCond} does not impose excitation requirements on the system trajectory, the excitation requirements are imposed on a user-selected set of points in the state space. Furthermore, Assumption \ref{ass:CLBCADPLearnCond} can be verified online. Since $\lambda_{\min}\left(\sum_{k=1}^{N}\frac{\omega_{k}\left(t\right)\omega_{k}^{T}\left(t\right)}{\rho_{k}^{2}\left(t\right)}\right)$ is non-decreasing in the number of samples, $N$, Assumption \ref{ass:CLBCADPLearnCond} can be met, heuristically, by increasing the number of samples.

\section{Stability Analysis}\label{Stability Analysis}
In the following theorem, the stability of the trajectories of the transformed system, and the estimation errors $\tilde{W}_{c}$, $\tilde{W}_{a}$, and $\tilde{\theta}$ are shown.  
\begin{theorem}\label{thm1}
    Provided Assumptions \ref{ass:finite_excitation}, \ref{ass:known_controller}, and \ref{ass:CLBCADPLearnCond} hold, the gains are selected large enough based on \eqref{eq:gain1} - \eqref{eq:gain4}, and the weights $\hat{\theta}$, $\hat{W}_{c}$, $\Gamma$, and $\hat{W}_{a}$ are updated according to \eqref{theta_update}, \eqref{W_c}, \eqref{gamma}, and \eqref{W_a}, respectively, then the estimation errors $\tilde{W}_{c}$, $\tilde{W}_{a}$, and $\tilde{\theta}$ and the trajectories of the transformed system in \eqref{eq:BTDynamics} under the controller in \eqref{u_command} are locally uniformly ultimately bounded. 
\end{theorem}
\begin{proof}
    See Theorem \ref{thm1} in Appendix.
\end{proof}
Using Lemma \ref{lem:trajectoryRelation}, it can then be concluded that the feedback control law
\begin{equation}\label{u_command_original}
    u(t)=\begin{cases} \psi\left(b_{(a,A)}(x(t)),t\right), & 0<t<T,\\
    \hat{u}\left(b_{(a,A)}(x(t)),\hat{W}_{a}(t)\right), & t\geq T,\end{cases}
\end{equation}
applied to the original system in \eqref{eq:Dynamics}, achieves the control objective stated in section \eqref{control object}.

\section{Simulation}\label{Simulation}
To demonstrate the performance of the developed method for a nonlinear system with an unknown value function, two simulation results, one for a two-state dynamical system \eqref{eq:sim1}, and one for a four-state  dynamical  system \eqref{eq:simrobot} corresponding to a two-link planar robot manipulator, are provided. 
\medskip
\subsection{Two state dynamical system}\label{simsec1}
The dynamical system is given by
\begin{equation}
    \dot {x} = f(x)\theta+g(x)u \label{eq:sim1}
\end{equation}
where \begin{equation} \label{sim_dyn}
f(x) = \begin{bmatrix}
     x_{2} & 0  & 0 & 0\\
     0 & x_{1} & x_{2} & x_{2}(\cos(2x_{1})+2)^2
\end{bmatrix},
\end{equation}
$\theta = \left[\theta_{1};\theta_{2}; \theta_{3}; \theta_{4}\right]$ and $g(x)= [0; \cos(2x_{1})+2]$.
The BT version of the system can be expressed in the form \eqref{eq:BTDynamics} with $G(s) =[0 ; G_{2_1}]$ and $
     y(s) = \begin{bmatrix}
     F_{1_1} & 0  & 0 & 0\\
     0 & F_{2_2} & F_{2_3} & F_{2_4}
\end{bmatrix}$, where 
\begin{gather*}
    F_{1_1} = B_{1}(s_{1})x_{2}, \quad F_{2_2} = B_{2}(s_{2})x_{1},
    F_{2_3} = B_{2}(s_{2})x_{2},\quad F_{2_4} = B_{2}(s_{2})x_{2}(\cos(2x_{1})+2)^2,\\
    G_{2_1} = B_{2}(s_{2})\cos(2x_{1})+2.
\end{gather*}
The state $x$ = $[x_{1} \  \ x_{2}]^{T}$ needs to satisfy the constraints $x_{1} \in (-7,5)$ and $x_{2} \in (-5,7)$. The objective for the controller is to minimize the infinite horizon cost function in \eqref{cost function}, with $Q = \text{diag}(10,10)$ and $R = 0.1$. The basis functions for value function approximation are selected as $\sigma(s) = [s_{1}^{2} ; s_{1}s_{2} ; s_{2}^{2}]$. The initial conditions for the system and the initial guesses for the weights and parameters are selected as $ x(0) =  [-6.5;6.5]$, $\hat{\theta}(0) = \left[0;0;0;0\right]$, $\Gamma(0)= \text{diag}(1,1,1)$, and $\hat{W}_{a}(0) = \hat{W}_{c}(0) = \left[\nicefrac{1}{2};\nicefrac{1}{2};\nicefrac{1}{2}\right] $. The ideal values of the unknown parameters in the system model are $\theta_{1} = 1$, $\theta_{2} = -1$, $\theta_{3} = -0.5$, $\theta_{4} = 0.5$, and the ideal values of the actor and the critic weights are unknown. The simulation uses 100 fixed Bellman error extrapolation points in a 4x4 square around the origin of the $s-$coordinate system.
\begin{table}[ht]
    \centering
    \caption{Comparison of costs for a single barrier transformed trajectory of \eqref{eq:sim1}, obtained using the optimal feedback controller generated via the developed method, and obtained using pseudospectral numerical optimal control software}
    \label{table_two_state_cost}
    \begin{tabular}{p{10cm}p{2cm}} 
        \hline 
        Method & Cost \\
        \hline 
        %MBRLBT with optimal weights and Nominal Gains  & 80.6730\\
        BT MBRL with FCL  & 71.8422\\ 
        GPOPS II (\cite{SCC.GPOPS})& 72.9005  \\
        \hline
    \end{tabular}
\end{table}

\subsubsection{Results for the two state system} \label{ Result_sim1}

As seen from Fig. \ref{fig:original_state}, the system state $x$ stays within the user-specified safe set while converging to the origin. The results in Fig. \ref{fig: Weight_EE} indicate that the unknown weights for both the actor and critic NNs converge to similar values. As demonstrated in Fig. \ref{fig: ParEE} the parameter estimation errors also converge to the zero.

Since the ideal actor and critic weights are unknown, the estimates cannot be directly compared against the ideal weights. To gauge the quality of the estimates, the trajectory generated by the controller $ u(t) = \hat{u}\left(s(t),\hat{W}_{c}^{*}\right), $ where $\hat{W}_{c}^{*}$ is the final value of the critic weights obtained in Fig. \ref{fig: Weight_EE}, starting from a specific initial condition, is compared against the trajectory obtained using an \emph{offline} numerical solution computed using the GPOPS II optimization software (see, e.g., \cite{SCC.GPOPS}). The total cost, generated by numerically integrating \eqref{cost function}, is used as the metric for comparison. The costs are computed over a finite horizon, selected to be roughly 5 times the time constant of the optimal trajectories. The results in Table \ref{table_two_state_cost} indicate that while the two solution techniques generate slightly different trajectories in the phase space (see Fig. \ref{fig: Comparison_Optimal}) the total cost of the trajectories is similar.
\begin{figure}
\centering
		\includegraphics[width=0.65\columnwidth]{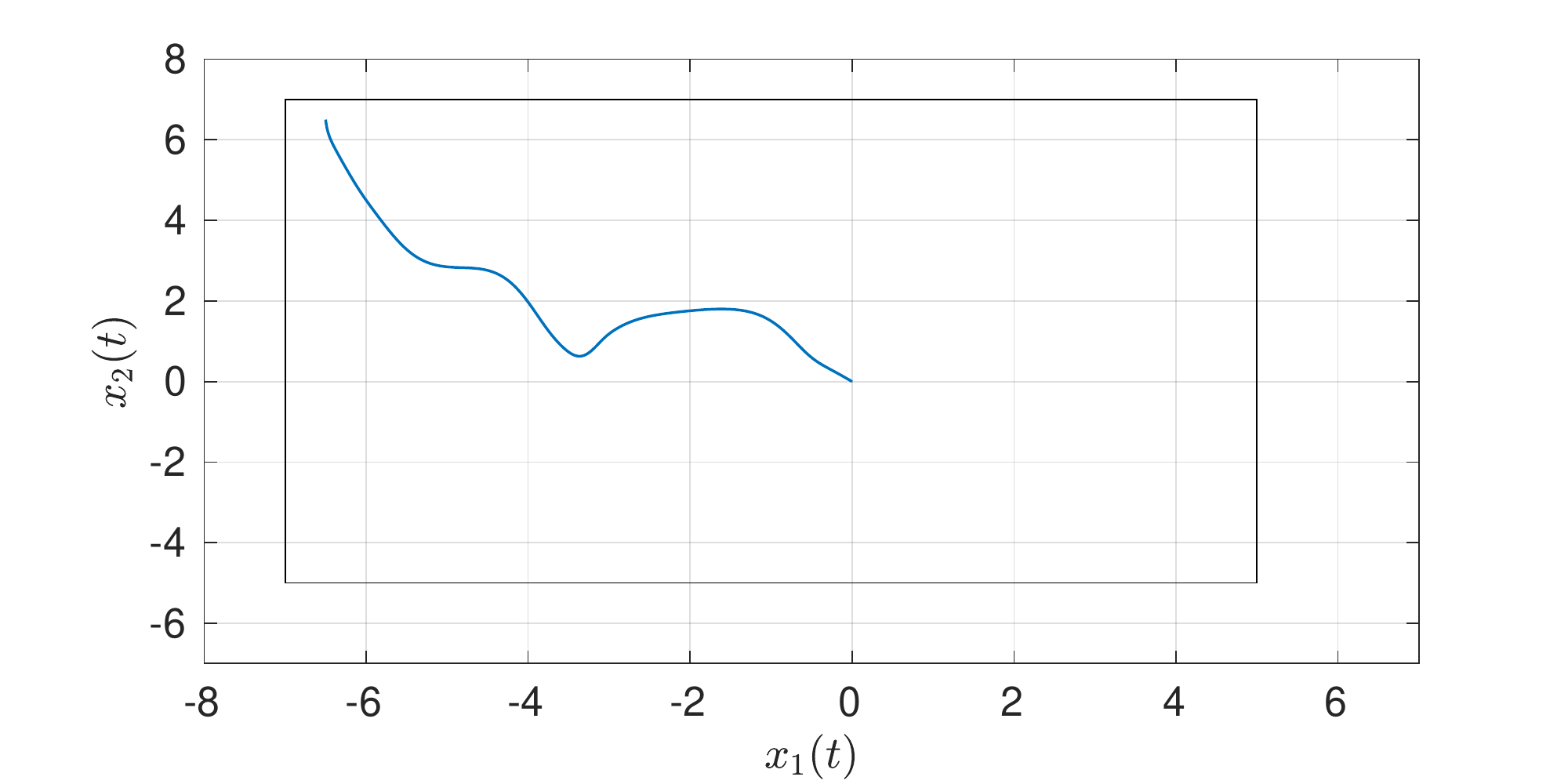}
		\caption{Phase portrait for the two-state dynamical system using MBRL with FCL in the original coordinates. The boxed area represents the user-selected safe set.}
		\label{fig:original_state}
\end{figure}
\begin{figure}
\centering
		\includegraphics[width=0.65\columnwidth]{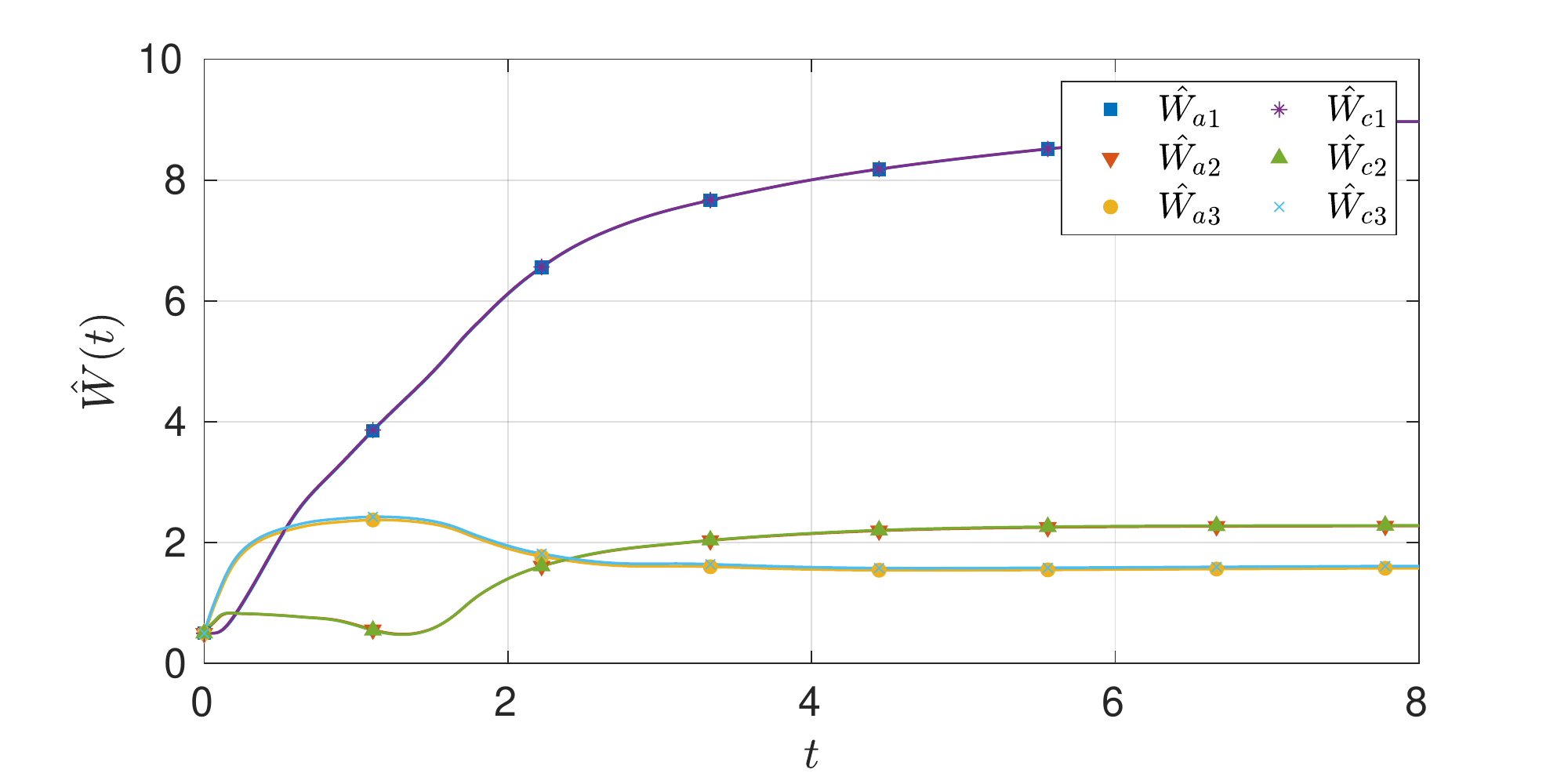}
		\caption{Estimates of the actor and the critic weights under nominal gains for the two-state dynamical system.}
		\label{fig: Weight_EE}
\end{figure}
\begin{figure}
\centering
		\includegraphics[width=0.65\columnwidth]{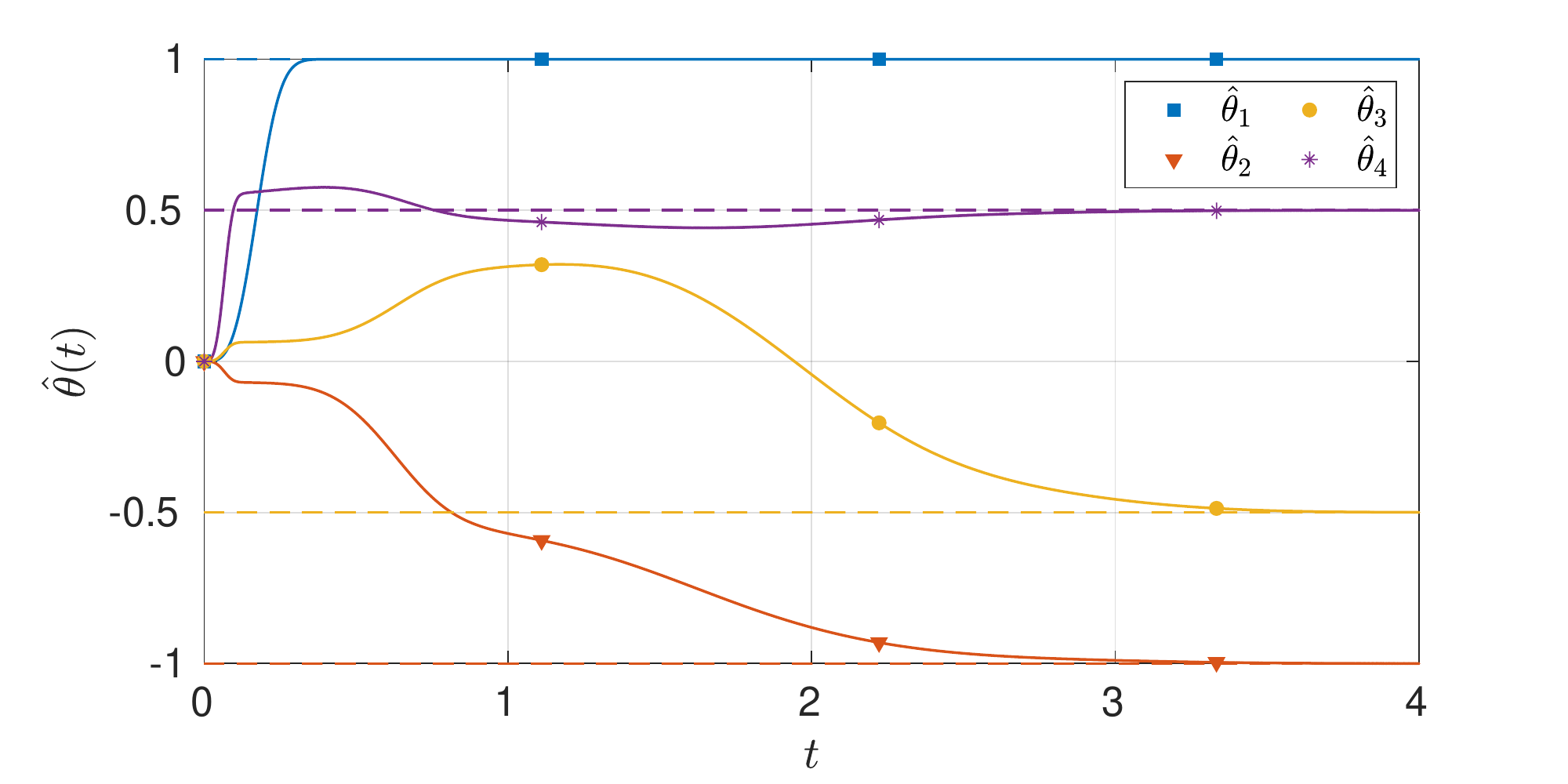}
		\caption{Estimates of the unknown parameters in the system under the nominal gains for the two-state dynamical system. The dash lines in the figure indicates the ideal values of the parameters.}
		\label{fig: ParEE}
\end{figure}
\begin{figure}
\centering
		\includegraphics[width=0.65\columnwidth]{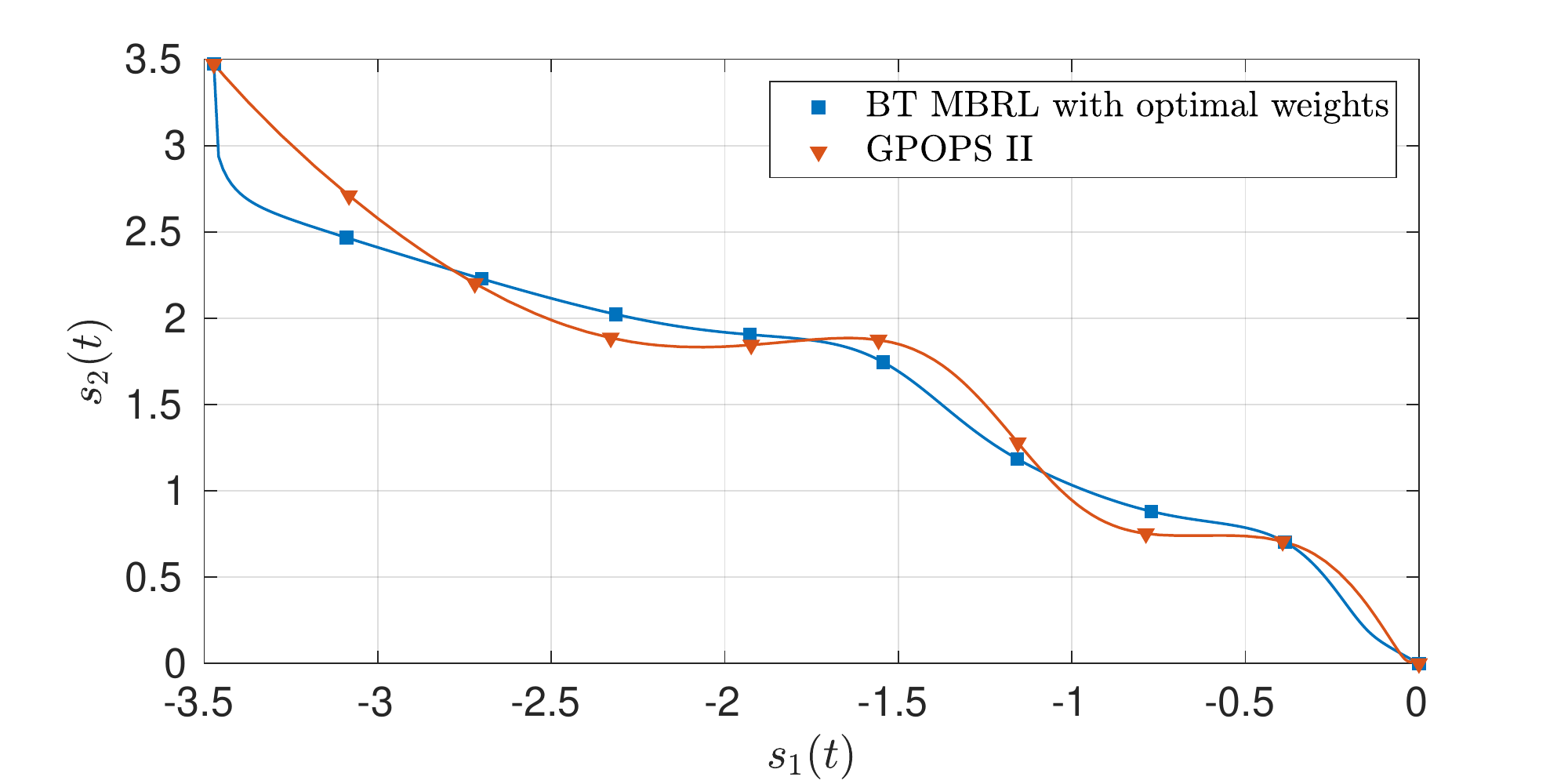}
		\caption{Comparison of the optimal trajectories obtained using GPOPS II and using BT MBRL with FCL and fixed optimal weights for the two-state dynamical system.}
		\label{fig: Comparison_Optimal}
\end{figure}

\subsubsection{Sensitivity Analysis for the two state system} \label{sen_sim1}
To study the sensitivity of the developed technique to changes in various tuning parameters, a one-at-a-time sensitivity analysis is performed. The parameters $k_{c1}$, $k_{c2}$, $k_{a1}$, $k_{a2}$, $\beta$, and $v$ are selected for the sensitivity analysis. The costs of the trajectories, under the optimal feedback controller obtained using the developed method, are presented in Table \ref{table_two_state_sensitivity} for 5 different values of each parameter. The parameters are varied in a neighborhood of the nominal values (selected through trial and error) $k_{c1} = 0.3$, $k_{c2} = 5$, $k_{a1} = 180$, $k_{a2} = 0.0001$, $\beta = .03$, and $v = 0.5$. The value of $\beta_{1}$ is set to be $\text{diag}(50,50,50,50)$. The results in Table \ref{table_two_state_sensitivity} indicate that the developed method is robust to small changes in the learning gains.
\begin{table}[ht]
    \centering
    \caption{Sensitivity Analysis for the two state system}
    \label{table_two_state_sensitivity}
    \begin{tabular}{p{2cm}p{2cm}p{2cm}p{2cm}p{2cm}p{2cm}p{2cm}} 
      \hline 
     $k_{c_1}$=  & 0.01 & 0.05 &  0.1 & 0.2  & 0.3    \\ 
       \hline 
     Cost   & 72.7174 & 72.6919 &  72.5378 & 72.3019  & 72.1559  \\ 
      \hline
         \hline 
       $k_{c_2}$=  & 2 &  3 & 5  & 10 & 15  \\ 
       \hline  
     Cost   &71.7476&72.3198&72.1559&71.8344&71.7293\\ 
      \hline
         \hline  
       $k_{a_1}$=  & 175 & 180 & 250  & 500 & 1000  \\ 
    \hline  
     Cost &72.1568 &  72.1559 &72.1384 & 72.1085 &72.0901  \\ 
      \hline
         \hline 
       $k_{a_2}$=  & 0.0001 & 0.0009  & 0.001  & 0.005 & 0.01 \\ 
       \hline  
     Cost   &72.1559&  72.1559&  72.1559&  72.1559&  72.1559\\ 
      \hline 
        \hline 
       $\beta$=  & 0.001 & 0.005 &  0.01 & 0.03   & 0.04   \\ 
      \hline 
     Cost   &72.2141&72.1559&72.1958&72.1559&72.1352 \\ 
      \hline
        \hline  
       $v$=  & 0.5 &  1 & 10  & 50 & 100 \\ 
       \hline 
     Cost &72.1559&72.4054&72.6582&79.1540&81.32  \\ 
      \hline
         \hline 
    \end{tabular}
\end{table}
\medskip
\subsection{Four state dynamical system}\label{simsec2} 
The four-state dynamical system corresponding to a two-link planar robot manipulator is given by
\begin{equation}
    \dot {x} = f_{1}(x)+ f_{2}(x)\theta+g(x)u \label{eq:simrobot}
\end{equation}
where \begin{equation} \label{sim_dyn1robot}
f_{1}(x) = 
\begin{bmatrix}
     x_{3}\\
     x_{4}\\
     -M^{-1}V_{m} \begin{bmatrix} x_{3}\\x_{4} \end{bmatrix}
\end{bmatrix}, \quad
f_{2}(x) = 
\begin{bmatrix}
     0, \  0, \  0, \  0\\
     0, \  0, \  0, \  0\\
     -[M^{-1}, M^{-1}]D 
\end{bmatrix}, \quad \theta = \begin{bmatrix} f_{d_{1}} \\ f_{d_{2}} \\ f_{s_{1}} \\ f_{s_{1}} \end{bmatrix},
\end{equation}
\begin{equation} \label{sim_dyn3robot}
    g(x) = 
    \begin{bmatrix}
     0,0 \\ 0,0 \\ (M^{-1})^{T}
    \end{bmatrix}, \quad D \coloneqq \mathrm{diag}
    \begin{bmatrix}
     x_{3}, x_{4}, \tanh(x_{3}), \tanh(x_{4})
    \end{bmatrix}, 
\end{equation}
\begin{equation} \label{M}
    M \coloneqq 
    \begin{bmatrix}
     p_{1}+2p_{3}c_{2} & p_{2}+p_{3}c_{2} \\ p_{2}+p_{3}c_{2} & p_{2}
    \end{bmatrix} \in \mathbb{R}^{2 \times 2}, \quad
    V_{M} \coloneqq 
    \begin{bmatrix}
     -p_{3}s_{2}x_{4} &  -p_{3}s_{2}(x_{3}+x_{4}) \\  p_{3}s_{2}x_{3} & 0
    \end{bmatrix} \in \mathbb{R}^{2 \times 2},  
\end{equation}
with $s_{2} = \sin(x_{2})$, $c_{2} = \cos(x_{2})$, $p_{1} = 3.473$, $p_{2} = 0.196$, $p_{3} = 0.242$. The positive constants $f_{d_{1}},f_{d_{2}},f_{s_{1}},f_{s_{1}} \in \mathbb{R}$ are the unknown parameters. The parameters are selected as $f_{d_{1}} = 5.3 , f_{d_{2}} = 1.1, f_{s_{1}} = 8.45, f_{s_{1}}= 2.35$. The state $x$ = $[x_{1} \  \ x_{2} \  \ x_{3} \  \ x_{4}]^{T}$ that corresponds to angular positions and the angular velocities of the two links needs to satisfy the constraints, $x_{1} \in (-7,5)$, $x_{2} \in (-7,5)$, $x_{3} \in (-5,7)$ and $x_{4} \in (-5,7)$. The objective for the controller is to minimize the infinite horizon cost function in \eqref{cost function}, with $Q = \text{diag}(1,1,1,1)$ and $R = \text{diag}(1,1)$ while identifying the unknown parameters $ \theta \in \mathbb{R}^4 $ that correspond to static and dynamic friction coefficients in the two links. The ideal values of the the unknown parameters are $\theta_{1} = 5.3$, $\theta_{2} = 1.1$, $\theta_{3} = 8.45$, and $\theta_{4} = 2.35$.
The basis functions for value function approximation are selected as $\sigma(s) = 
[s_{1}s_{3} ;s_{2}s_{4} ;s_{3}s_{2} ;s_{4}s_{1} ;s_{1}s_{2} ;s_{4}s_{3} ;s_{1}^{2} ; s_{2}^{2}; s_{3}^{2} ;s_{4}^{2}]$. The initial conditions for the system and the initial guesses for the weights and parameters are selected as $ x(0) =  [-5;-5;5;5]$, $\hat{\theta}(0) = \left[5;5;5;5\right]$, $\Gamma(0)= \text{diag}(10,10,10,10,10,10,10,10,10,10)$, and $\hat{W}_{a}(0) = \hat{W}_{c}(0) = \left[60;2;2;2;2;2;40;2;2;2\right]$. The ideal values of the actor and the critic weights are unknown. The simulation uses 100 fixed Bellman error extrapolation points in a 4x4 square around the origin of the $s-$coordinate system.

\subsubsection{Results for the four state system} \label{Result_sim2}
As seen from Fig. \ref{fig:original_state_robot}, the system state $x$ stays within the user-specified safe set while converging to the origin. As demonstrated in Fig. \ref{fig: ParEE_robot}, the parameter estimations converge to the true values. A comparison with offline numerical optimal control, similar to the procedure used for the two-state, yields the results in Table \ref{table_four_state_cost} indicate that the two solution techniques generate slightly different trajectories in the state space (see Fig. \ref{fig: Comparison_Optimal_robot}) and the total cost of the trajectories is different. We hypothesize that the difference in costs is due to the basis for value function approximation being unknown.
\begin{figure}
\centering
		\includegraphics[width=0.65\columnwidth]{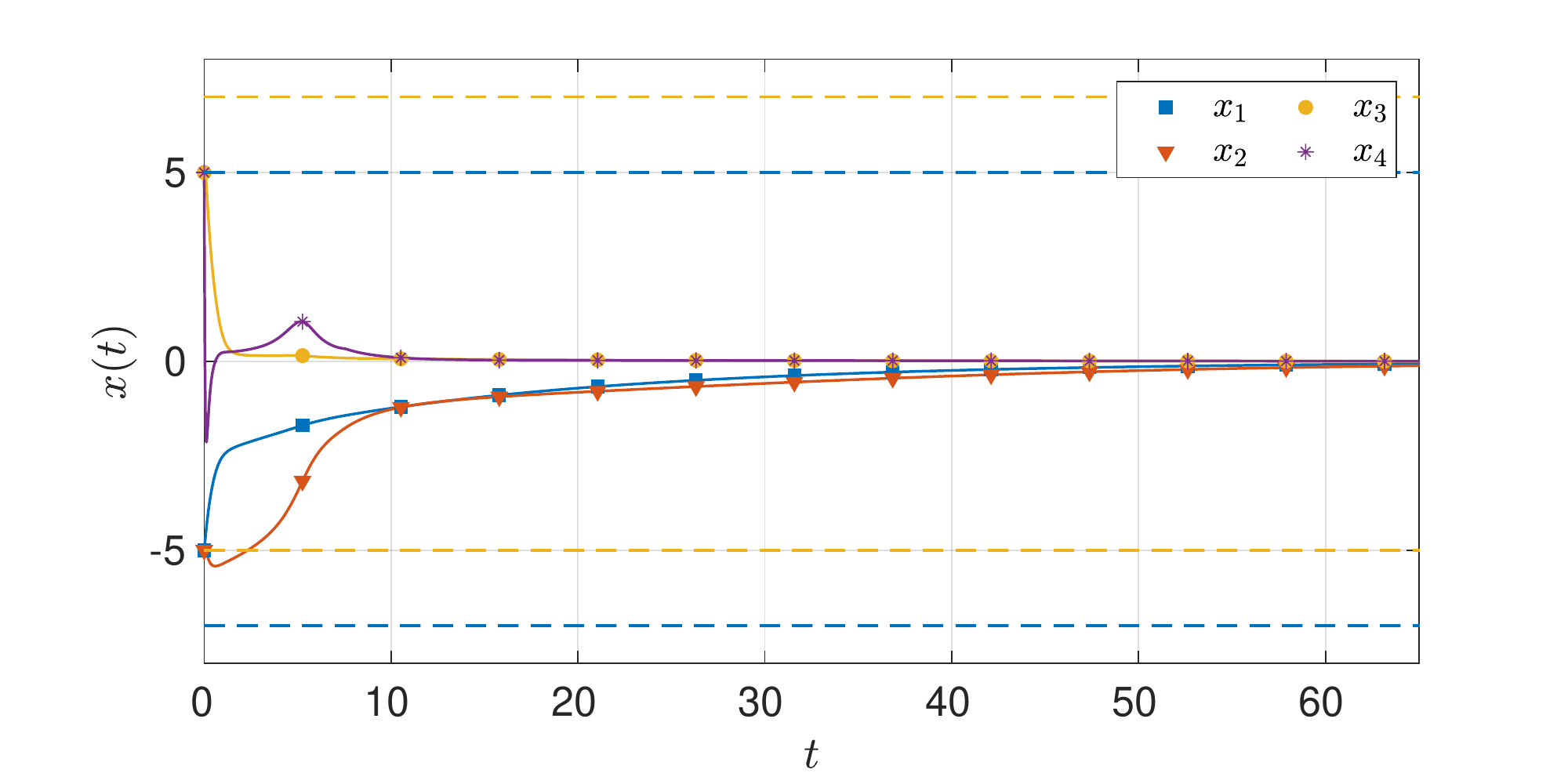}
		\caption{ State trajectories for the four-state dynamical system using MBRL with FCL in the original coordinates. The dash lines represent the user-selected safe set.}
		\label{fig:original_state_robot}
\end{figure}
In summary, the newly developed method can achieve online optimal control thorough a BT MBRL approach while estimating the value of the unknown parameters in the system dynamics and ensuring safety guarantees in the original coordinates during the learning phase.
\begin{figure}
\centering
		\includegraphics[width=0.65\columnwidth]{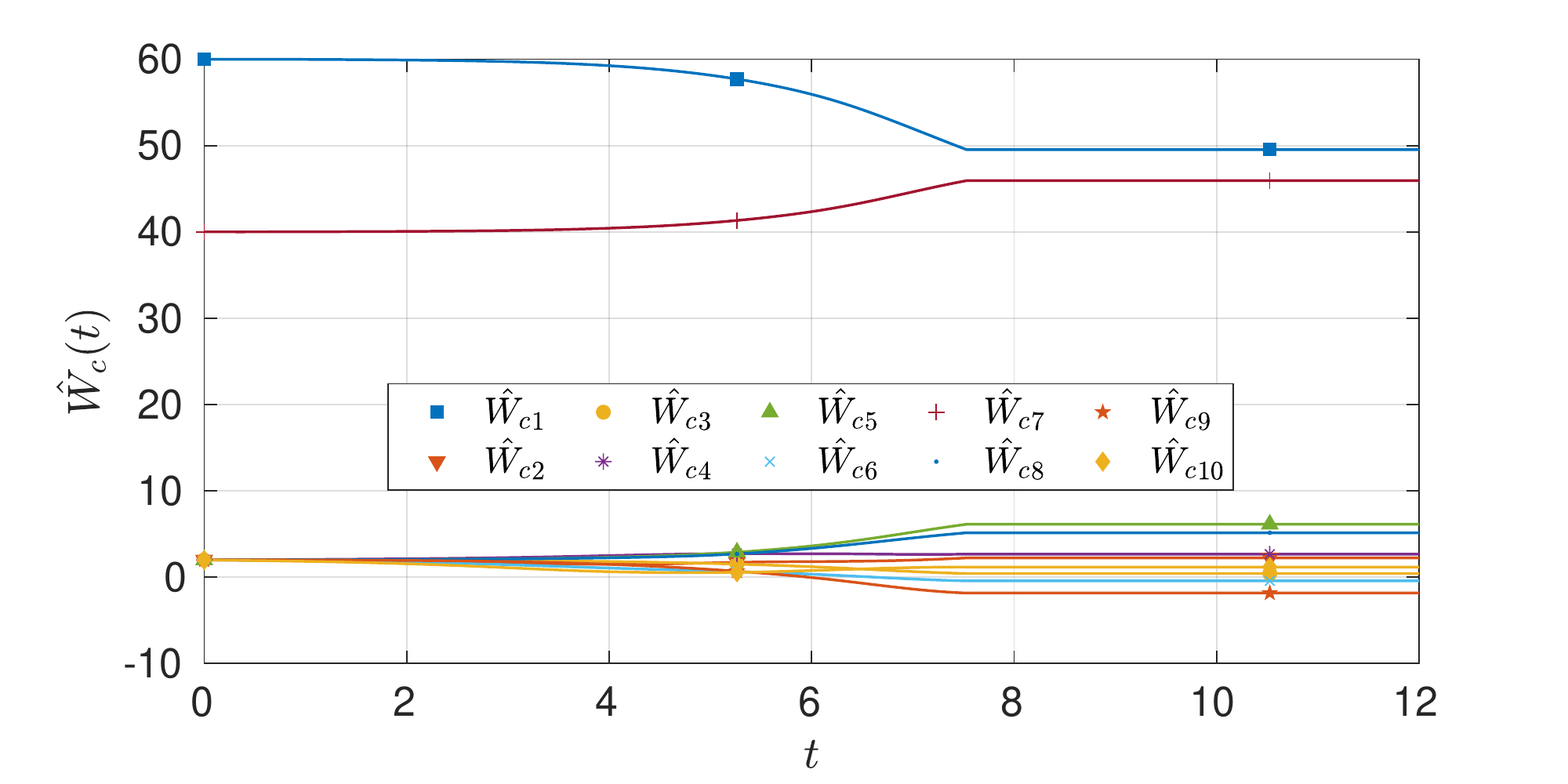}
		\caption{Estimates of the critic weights under nominal gains for the four-state dynamical system.}
		\label{fig: Weight_EE_robot}
\end{figure}
\begin{figure}
\centering
		\includegraphics[width=0.65\columnwidth]{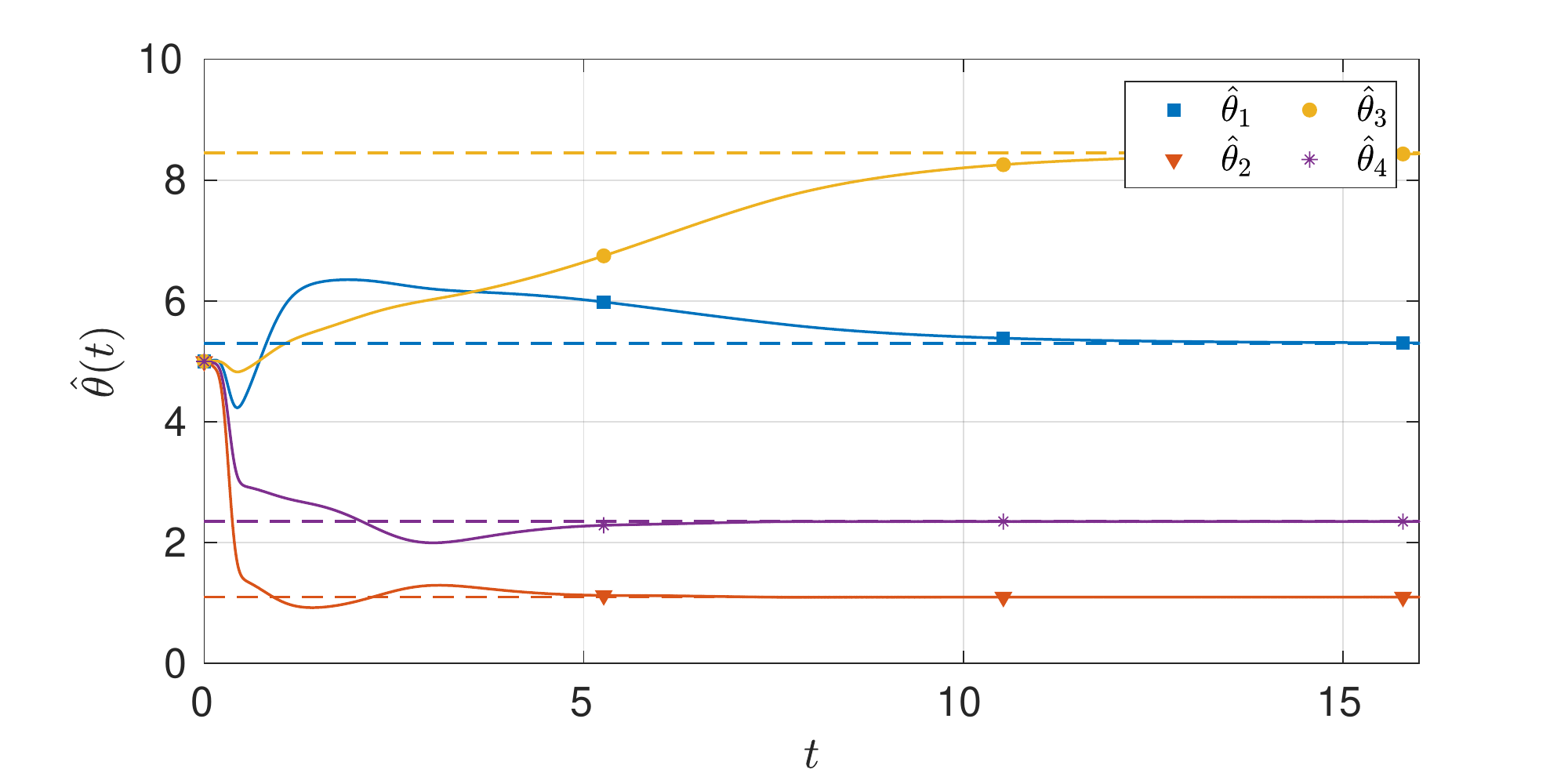}
		\caption{Estimates of the unknown parameters in the system under the nominal gains for the four-state dynamical system. The dash lines in the figure indicates the ideal values of the parameters.}
		\label{fig: ParEE_robot}
\end{figure}
\begin{table}
    \centering
    \caption{Costs for a single barrier transformed trajectory of \eqref{eq:simrobot}, obtained using the developed method, and using pseudospectral numerical optimal control software}
    \label{table_four_state_cost}
    \begin{tabular}{p{10cm}p{2cm}} 
      \hline 
      Method & Cost \\
      \hline 
        BT MBRL with FCL  &  95.1490\\ 
        GPOPS II & 57.8740  \\
      \hline
    \end{tabular}
\end{table}
The following section details a one-at-a-time sensitivity analysis and study the sensitivity of the developed technique to changes in various tuning parameters.

\subsubsection{Sensitivity Analysis for the four state system}
The parameters $k_{c1}$, $k_{c2}$, $k_{a1}$, $k_{a2}$, $\beta$, and $v$ are selected for the sensitivity analysis. The costs of the trajectories, under the optimal feedback controller obtained using the developed method, are presented in Table \ref{table_four_state_sensitivity} for 5 different values of each parameter.
\begin{table}
    \centering
    \caption{Sensitivity Analysis for the four state system}
    \label{table_four_state_sensitivity}
    \begin{tabular}{p{2cm}p{2cm}p{2cm}p{2cm}p{2cm}p{2cm}p{2cm}} 
     \hline  
    $k_{c_1}$=  & 0.01 & 0.05 &  0.1 & 0.5  & 1    \\
     \hline  
     Cost   & 95.91 & 95.4185 &  95.1490 & 94.1607& 93.5487  \\ 
      \hline
      \hline  
       $k_{c_2}$=  & 1 &  5 & 10  & 20 & 30 \\
     \hline  
     Cost   &304.4&101.0786&95.1490&92.7148&93.729\\
      \hline
      \hline 
       $k_{a_1}$=  & 5 & 10 & 20  & 30 & 50  \\
    \hline 
     Cost &94.9464 &  95.1224 & 95.1490 & 95.1736 &95.1974 \\
      \hline
       \hline  
       $k_{a_2}$=  & 0.05 & 0.1  & 0.2  & 0.5 & 1 \\
    \hline 
     Cost   &95.2750&  95.2480&  95.1490&  94.9580&  94.6756\\
      \hline 
      \hline 
       $\beta$=  & 0.1 & 0.5  & 0.8 & 0.9 & 0.95 \\
     \hline  
     Cost   &125.33&109.7721&95.1490&92.91&93.7231\\
      \hline
      \hline 
       $v$=  & 50 &  70 & 100  & 125 & 150\\ 
     \hline 
     Cost &92.2836&93.34&95.1490&96.1926&97.9870 \\
      \hline
       \hline
    \end{tabular}
\end{table}
The parameters are varied in a neighborhood of the nominal values (selected through trial and error) $k_{c1} = 0.1$, $k_{c2} = 10$, $k_{a1} = 20$, $k_{a2} = 0.2$, $\beta = 0.8$, and $v = 100$. The value of $\beta_{1}$ is set to be $\text{diag}(100,100,100,100)$. The results in Table \ref{table_four_state_sensitivity} indicate that the developed method is not sensitive to small changes in the learning gains.

\begin{figure}
     \centering	
     \includegraphics[width=0.65\columnwidth]{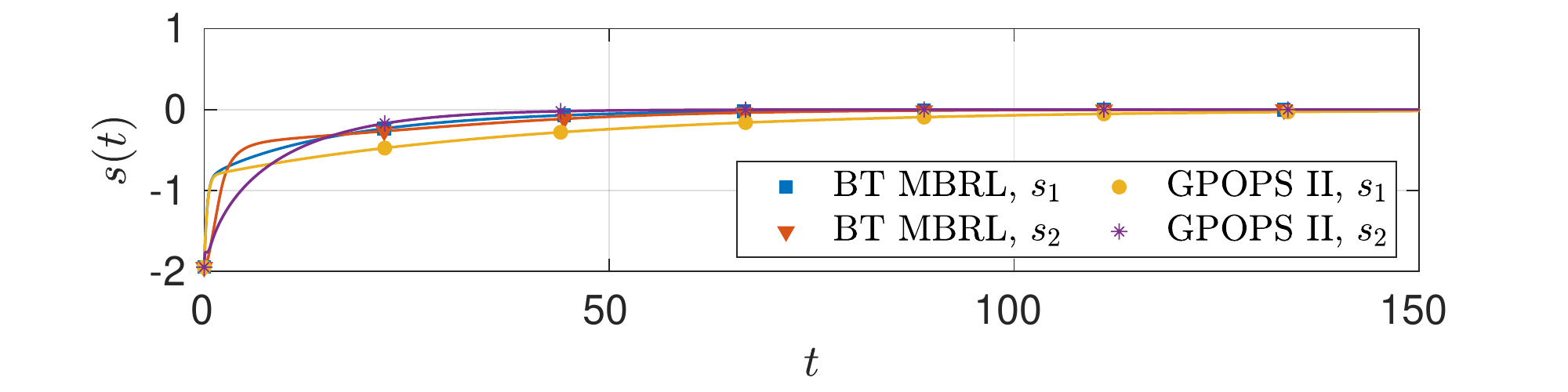}
     
     \includegraphics[width=0.65\columnwidth]{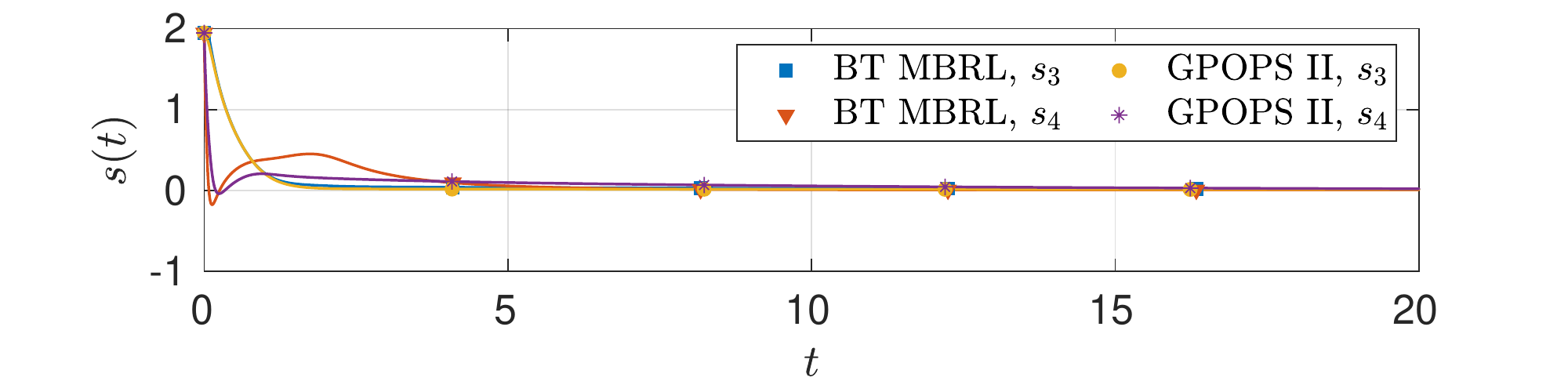}
     
			\caption{Comparison of the optimal angular position (top) and angular velocity (bottom) trajectories obtained using GPOPS II and BT MBRL with fixed optimal weights for the four-state dynamical system.}
  	\label{fig: Comparison_Optimal_robot}
\end{figure}

The results in Tables 2 and 4 indicate that the developed method is not sensitive to small changes in the learning gains. While reduced sensitivity to gains simplifies gain selection, as indicated by the local stability result, the developed method is sensitive to selection of basis function and initial guesses of the unknown weights. Due to high dimensionality of the vector of unknown weights, a complete characterization of the region of attraction is computationally difficult. As such, the basis functions and the initial guess were selected via trial and error.

\section{Conclusion}\label{Conclusion}
This paper develops a novel online safe control synthesis technique which relies on a nonlinear coordinate transformation that transforms a constrained optimal control problem into an unconstrained optimal control problem. A model of the system in the transformed coordinates is simultaneously learned and utilized to simulate experience. Simulated experience is used to realize convergent RL under relaxed excitation requirements. Safety of the closed-loop system, expressed in terms of box constraint, regulation of the system states to a neighborhood of the origin, and convergence of the estimated policy to a neighborhood of the optimal policy in transformed coordinates is established using a Lyapunov-based stability analysis.

While the main result of the paper states that the state is uniformly ultimately bounded, the simulation results hint towards asymptotic convergence of the part of the state that corresponds to the system trajectories, $x(\cdot)$. Proving such a result is a part of future research.

Limitations and possible extensions of the ideas presented in this paper revolve around the two key issues: (a) safety, and (b) online learning and optimization. The barrier function used in the BT to address safety can only ensure a fixed box constraint. A more generic and adaptive barrier function, constructed, perhaps, using sensor data is a subject for future research.

For optimal learning, parametric approximation techniques are used to approximate the value functions in this paper. Parametric approximation of the value function requires selection of appropriate basis functions which may be hard to find for the barrier-transformed dynamics. Developing techniques to systematically determine a set of basis functions for real-world systems is a subject for research.

The barrier transformation method to ensure safety relies on knowledge of the dynamics of the system. While this paper addresses parametric uncertainties, the BE method could potentially result in a safety violation due to unmodeled dynamics. In particular, the safety guarantees developed in this paper rely on the relationship (Lemma \ref{lem:trajectoryRelation}) between trajectories of the original dynamics and the transformed system, which holds in the presence of parametric uncertainty, but fails if a part of the dynamics is not included in the original model. Further research is needed to establish safety guarantees that are robust to unmodeled dynamics (for a differential games approach to robust safety, see \cite{SCC.Yang.Ding.ea2020}).

\pagebreak
\bibliographystyle{IEEETrans.bst}
\bibliography{scc,sccmaster,scctemp}
\pagebreak
\begin{appendix}

\begin{manuallemma}{\ref{lem:trajectoryRelation}}
    If $t \mapsto \Phi\big(t,b(x^{0}),\zeta\big)$ is a complete Carath\'{e}odory solution to \eqref{eq:BTDynamics}, starting from the initial condition $b(x^{0})$, under the feedback policy $(s,t) \mapsto \zeta (s,t)$ and $t \mapsto \Lambda(t,x^{0},\xi)$ is a Carath\'{e}odory solution to \eqref{eq:Dynamics}, starting from the initial condition $x^{0}$, under the feedback policy $(x,t) \mapsto \xi(x,t)$, defined as $\xi(x,t) = \zeta(b(x),t)$, then $\Lambda(\cdot,x^{0}, \xi)$ is complete and $ \Lambda(t,x^{0}, \xi) = b^{-1}\left(\Phi(t,b(x^{0}),\zeta)\right) $ for all $t \in \mathbb{R}_{\geq 0}$.
\end{manuallemma}
\begin{proof}
	Note that since $t \mapsto \Phi\big(t,b(x^{0}),\zeta\big)$ is a complete Carath\'{e}odory solution to $ \dot{s} = F(s) + G(s) \zeta(s,t) $, it is differentiable at almost all $ t \in \mathbb{R}_{\geq 0} $. Since $ b^{-1} $ is smooth, $ t \mapsto b^{-1}\left(\Phi\left(t,b(x^{0}),\zeta\right)\right) $ is also differentiable at almost all $ t \in \mathbb{R}_{\geq 0} $. That is,
	\begin{align*}
		\frac{\mathrm{d}(b^{-1}\circ\Phi_i)}{\mathrm{d}t}\left(t,b(x^{0}),\zeta\right) =\frac{\mathrm{d} b^{-1}_{(a_i,A_i)}(y)}{\mathrm{d} y}|_{y=\Phi_{i}\left(t,b\left(x^{0}\right), \zeta\right)} \frac{\mathrm{d} \Phi_{i}}{\mathrm{d} t}\left(t,b(x^{0}),\zeta\right), 
	\end{align*}
	for almost all $t \in \mathbb{R}_{\geq 0}$ and all $i=1,\cdots,n$, where $\Phi_i$ denotes the $i$th component of $\Phi$. As a result,
	\begin{align*}
		\frac{\mathrm{d}(b^{-1}\circ\Phi_i)}{\mathrm{d}t}\left(t,b(x^{0}),\zeta\right) =  \frac{\left(F\left(\Phi\left(t,b(x^{0}),\zeta\right)\right)\right)_i}{B_i\left(\Phi_i\left(t,b(x^{0}),\zeta\right)\right)} 
		+ \frac{\left(G\left(\Phi\left(t,b(x^{0}),\zeta \right)\right)\right)_i \zeta\left(\Phi\left(t,b(x^{0}),\zeta\right),t\right)}{B_i\left(\Phi_i\left(t,b(x^{0}),\zeta\right)\right)},
    \end{align*}
		for almost all $ t \in \mathbb{R}_{\geq 0} $ and all $i=1,\hdots,n$. By the construction of $ F $, $ G $, and $ \xi $,
	\begin{multline*}
		\frac{\mathrm{d}(b^{-1}\circ\Phi)}{\mathrm{d}t}\left(t,b(x^{0}),\zeta\right)  = f\left(b^{-1}\circ\Phi\left(t,b(x^{0}),\zeta\right)\right)\theta\\+g\left(b^{-1}\circ\Phi\left(t,b(x^{0}),\zeta\right)\right) \xi\left(b^{-1}\circ\Phi\left(t,b(x^{0}),\zeta\right),t\right),
	\end{multline*}
	for almost all $ t \in \mathbb{R}_{\geq 0} $. Clearly $t \mapsto b^{-1}\circ\Phi\left(t,b(x^{0}),\zeta\right)$ is a Carath\'{e}odory solution of \eqref{eq:Dynamics} on $\mathbb{R}_{\geq 0}$, starting from the initial condition $ b^{-1}\big(b(x^{0})\big)= x^{0}$ under the feedback policy $ (x,t)\mapsto \xi(x,t) $. By uniqueness of solutions $ \dot{x} = f(x)\theta + g(x) \xi(x,t) $ (which follows from local Lipschitz continuity of $ f $, $ g $, and $ b $ inside the barrier), $\Lambda(\cdot,x^0,\xi)$ is complete and $\Lambda(t,x^0,\xi) = b^{-1}\left(\Phi\left(t,b(x^{0}),\zeta\right)\right)$ for all $ t \in \mathbb{R}_{\geq 0} $.
\end{proof}

\begin{manuallemma}{\ref{existence_Caratheodory}}
    If $\|Y_{f}\|$ is non-decreasing in time then \eqref{z_dot} admits Carath\'{e}odory solutions.
\end{manuallemma}
\begin{proof}
    Since $\|Y_{f}(0)\| = 0,$ given any piecewise continuous control signal $t \mapsto u(t)$ and initial conditions $s^{0}$ and $\theta^{0}$, the Cauchy problem $\dot{z} = h_{1}(z,u)$, $z(0) = z^{0} = [s^{0};0;0;0;0;\theta^{0}]$ admits a unique Carath\'{e}odory solution $t \mapsto z_{1}(0,z^{0})$ over $[0,t^{*})$, with $t^{*} = \min (t_{1},t_{2})$, where $t_{1}$ = $\inf \{t \in \mathbb{R}_{\geq 0} \ |  \ \|Y_{f1}(t,z^{0}\| = \overline{Y_{f}}\}$ and $t_{2}$ = $\inf\{t \in \mathbb{R}_{\geq 0} | \lim_{\tau \mapsto t} \|z_{1}(\tau,z^{0})\| = \infty$\}, where $Y_{f1}$ denotes the $Y_{f}$ component of $z_{1}$.
    
    Given any $(t',z') \in \mathbb{R}_{\geq 0} \times \mathbb{R}^{2n+2p+p^{2}+np}$, the Cauchy problem $\dot{z} = h_{2}(z,u), z(t') = z'$, also admits a unique Carath\'{e}odory solution $t \mapsto z_{2}(t;t',z')$ over $[t',t^{**})$ where $t^{**}$ = $\min\Big(\infty,\big(\inf\{t \in \mathbb{R}_{\geq t'} | \lim_{\tau \mapsto t} \|z_{2}(\tau,b',z')\| = \infty\}\big)\Big)$.
    
    If $t^{*} = t_{2}$ then $t \mapsto z_{1}(t,z^{0})$ is also a unique Carath\'{e}odory solution to the Cauchy problem $\dot{z}= h(z,u)$, $z(0)=z^{0}$. If not, then
    \begin{equation*}
        t \mapsto z^{*}(t,z^{0}) = \begin{cases}z_{1}(t,z^{0}), \quad t < t_{1}, \\ z_{2}\big(t,t_{1}, \lim_{\tau \uparrow t_{1}} z_{1}(\tau, z^{0})\big), &t \geq t_{1},\end{cases} 
    \end{equation*} is a unique Carath\'{e}odory solution to the Cauchy problem $\dot{z}= h(z,u)$, $z(0)=z^{0}$.
\end{proof}
\begin{manualtheorem}{\ref{thm1}}
    Provided Assumptions \ref{ass:finite_excitation}, \ref{ass:known_controller}, and \ref{ass:CLBCADPLearnCond} hold, the gains are selected large enough based on \eqref{eq:gain1} - \eqref{eq:gain4}, and the weights $\hat{\theta}$, $\hat{W}_{c}$, $\Gamma$, and $\hat{W}_{a}$ are updated according to \eqref{theta_update}, \eqref{W_c}, \eqref{gamma}, and \eqref{W_a}, respectively, then the estimation errors $\tilde{W}_{c}$, $\tilde{W}_{a}$, and $\tilde{\theta}$ and the trajectories of the transformed system in \eqref{eq:BTDynamics} under the controller in \eqref{u_command} are locally uniformly ultimately bounded.
\end{manualtheorem}
\begin{proof}
Under Assumption \ref{ass:finite_excitation}, the state trajectories are bounded over the interval $ [0,T) $. Over the interval $[T,\infty)$, let $B_{r}\subset\mathbb{R}^{n+2L+p}$ denote a closed ball with radius $r$ centered at the origin. Let $\chi$ denote the projection of $B_{r}$ onto $\mathbb{R}^{n}$. For any continuous function $h:\mathbb{R}^{n}\to\mathbb{R}^{m}$, let the notation $\overline{\left\Vert \left(\cdot\right)\right\Vert }$ be defined as $\overline{\left\Vert h\right\Vert }\coloneqq\sup_{s^{o}\in\chi}\left\Vert h\left(s^{o}\right)\right\Vert$. To facilitate the analysis, let $\left\{ \varpi_{j}\in\mathbb{R}_{>0}\mid j=1,\cdots,7\right\}$ be constants such that $\varpi_{1}+\varpi_{2}+\varpi_{3}=1$, and $\varpi_{4}+\varpi_{5}+\varpi_{6}+\varpi_{7}=1$. Let $\underline{c}\in\mathbb{R}_{>0}$ be a constant defined as
\begin{equation} \label{c1}
\underline{c}\coloneqq\frac{\beta}{2\overline{\Gamma}k_{c2}}+\frac{\underline{c}_{3}}{2},
\end{equation}
$k_{5}$ be a positive constant defined as $k_{5} \coloneqq \bar{W}K_{c1}\overline{\|\nabla_{s}\sigma\|} L_{y}$.
and let $\iota\in\mathbb{R}$ be a positive constant defined as
\begin{multline} \label{iota}
\iota\triangleq\frac{\left(k_{c1}+k_{c2}\right)^{2}\overline{\left\Vert \Delta\right\Vert }^{2}}{4k_{c2}\underline{c}\varpi_{3}}+\frac{1}{4}\overline{\left\Vert G_{\epsilon}\right\Vert }  + \frac{1}{{4\left(k_{a1}+k_{a2}\right)\varpi_{6}}}   \left(\frac{1}{2}\overline{W}\overline{\left\Vert G_{\sigma}\right\Vert }+\frac{1}{2}\overline{\left\Vert \nabla_{s}\epsilon G^{T}\nabla_{s}\sigma^{T}\right\Vert}\right) \\ + \frac{1}{{4\left(k_{a1}+k_{a2}\right)\varpi_{6}}} \left(k_{a2}\overline{W} +\frac{1}{4}\left(k_{c1}+k_{c2}\right)\overline{W}^{2}\overline{\left \Vert G_{\sigma}\right \Vert }\right)^{2}. 
\end{multline}
To facilitate the stability analysis, let $V_{L}:\mathbb{R}^{n+2L+p}\times\mathbb{R}_{\geq 0} \to \mathbb{R}_{\geq 0}$ be a continuously differentiable candidate Lyapunov function defined as 
\begin{equation}\label{v_L0}
V_{L}\left(Z,t\right)\coloneqq V^{*}(s)+\frac{1}{2}\tilde{W}_{c}^{T}\Gamma^{-1}(t)\tilde{W}_{c}+\frac{1}{2}\tilde{W}_{a}^{T}\tilde{W}_{a}+V_{1}(\tilde{\theta}),
\end{equation} where $V^{*}$ is the optimal value function, $V_{1}$ was introduced in section \ref{para}  and $Z\triangleq\left[s;\:\tilde{W}_{c};\:\tilde{W}_{a};\:\tilde{\theta}\right]$. 
The update law in \eqref{W_c} ensures that the adaptation gain matrix is bounded such that
\begin{equation}
    \underline{\Gamma} \leq \|\Gamma (t) \| \leq \overline{\Gamma}, \forall t \in \mathbb{R}_{\geq T}.
\end{equation}
Using the fact that $V^{*} \ \text{and} \ V_{1}$ are positive definite, Lemma 4.3 from \cite{SCC.Khalil2002} yields 
\begin{equation} \label{eq:v_l}    \underline{v_{l}}\left(\left\Vert Z\right\Vert \right)\leq V_{L}\left(Z,t\right)\leq\overline{v_{l}}\left(\left\Vert Z\right\Vert \right),
\end{equation}
for all $t\in\mathbb{R}_{\geq T}$ and for all $Z\in\mathbb{R}^{n+2L+p}$, where $\underline{v_{l}},\overline{v_{l}}:\mathbb{R}_{\geq0}\rightarrow\mathbb{R}_{\geq0}$ are class $\mathcal{K}$ functions. 
Let $v_{l}:\mathbb{R}_{\geq0}\rightarrow\mathbb{R}_{\geq0}$ be a  function defined as 
$v_{l}\left(\left\Vert Z\right\Vert \right) \coloneqq \frac{\lambda_{\min}\{{Q\}}\|s\|^{2}}{2}+\frac{k_{c2}\underline{c}\varpi_{1}}{2}\left\Vert \tilde{W}_{c}\right\Vert ^{2}+\frac{\left(k_{a1}+k_{a2}\right)\varpi_{4}}{2}\left\Vert \tilde{W}_{a}\right\Vert ^{2}+\frac{\left\Vert \tilde{\theta}\right\Vert ^{2}}{2}.$

The sufficient conditions for ultimate boundedness of $Z$ are derived based on the subsequent stability analysis as
\begin{align}\label{eq:gain1}
    \bigg(k_{c2}\underline{c}\varpi_{2} -\frac{k_{5} r\epsilon}{2}\bigg)(k_{a1}+k_{a2})\varpi_{5}  \geq   \bigg(k_{a1}  + \frac{1}{4}\left(k_{c1}+k_{c2}\right)\overline{W}\overline{\left\Vert G_{\sigma}\right\Vert }\bigg),
\end{align}
\begin{equation}\label{eq:gain2}
    \left(k_{a1}+k_{a2}\right)\varpi_{7}\geq\frac{1}{4}\left(k_{c1}+k_{c2}\right)\overline{W}\overline{\left\Vert G_{\sigma}\right\Vert },
\end{equation}
\begin{equation}\label{eq:gain3}
    \lambda_{\min}\{Y_{f}(T)\} \geq \frac{k_{5}r}{2\epsilon}+1,
\end{equation}
\begin{equation}\label{eq:gain4}
    v_{l}^{-1}(\iota) < \overline{v_{l}}^{-1}(\underline{v_{l}}(r)).
\end{equation}
The bound on the function $F$ and the NN function approximation errors depend on the underlying compact set; hence, $\iota$ is a function of $r$. Even though, in general, $\iota$ increases with increasing $r$, the sufficient condition in \eqref{eq:gain4} can be satisfied provided the points for BE extrapolation are selected such that the constant $\underline{c}$, introduced in \eqref{c1} is large enough and that the basis for value function approximation are selected such that $\overline{\left\Vert \epsilon\right\Vert }$ and $\overline{\left\Vert \nabla\epsilon\right\Vert }$ are small enough.

The differential equation \eqref{eq:BTDynamics}, under the controller in \eqref{u_command}, along with \eqref{theta_update}, \eqref{W_c}, and \eqref{W_a}, constitute the closed-loop system $ \dot{Z} = h(Z,t) $ to be analyzed. Let $\dot{V}_L$ denote the orbital derivative of \eqref{v_L0} along the trajectories of the closed-loop system, i.e., $\dot{V}_L := \nabla_t V_L + \nabla_Z V_L(Z,t)h(Z,t)$. Then,
\begin{align}\label{lyapunovfunctionderivation}
    \dot{V}_L = \nabla_{s}V^{*}F+\nabla_{s}V^{*}G\hat{u}+{\tilde{W}}_{c}^{T}\Gamma^{-1}\dot{\tilde{W}}_{c}+\frac{1}{2}\tilde{W}_{c}^{T}\dot{\Gamma}^{-1}{\tilde{W}}_{c}+{\tilde{W}}_{a}^{T}\dot{\tilde{W}}_{a}+\dot{V}_{1}.
\end{align}
Substituting \eqref{W_c} - \eqref{W_a} in \eqref{lyapunovfunctionderivation} yields
\begin{multline}\label{lyapunovfunctionderivation11}
    \dot{V}_{L} \leq \nabla_{s}V^{*}\left(F+Gu^{*}\right)-\nabla_{s}V^{*}Gu^{*}+\nabla_{s}V^{*}G\hat{u}-\tilde{W}_{c}^{T}\Gamma^{-1}\bigg(-k_{c1}\Gamma\frac{\omega}{\rho}{\hat{\delta}}-\frac{1}{N}\Gamma\sum_{k=1}^{N}\frac{k_{c2}\omega_{i}}{\rho_{k}}{\hat{\delta}}_{k}\bigg) \\	-\frac{1}{2}\tilde{W}_{c}^{T}\Gamma^{-1}\bigg(\beta\Gamma-k_{c1}(\Gamma\frac{\omega\omega^{T}}{\rho^{2}}\Gamma)-\frac{k_{c2}}{N}\Gamma\sum_{k=1}^{N}\frac{\omega_{k}\omega_{k}^{T}}{\rho_{k}^{2}}\Gamma\bigg)\Gamma^{-1}\tilde{W}_{c} \\
	-\tilde{W}_{a}^{T}\bigg(-k_{a1}(\hat{W}_{a}-\hat{W}_{c})-k_{a2}\hat{W}_{a}+\big((\frac{k_{c1}\omega}{4\rho}\hat{W}_{a}^{T}G_{\sigma}+\sum_{k=1}^{N}\frac{k_{c2}\omega_{k}}{4N\rho_{k}}\hat{W}_{a}^{T}G_{\sigma k})^{T}\hat{W}_{c}\big)\bigg)+\dot{V}_{1}.
\end{multline}

The Lyapunov derivative can be rewritten as 
\begin{multline}\label{lyapunovfunctionderivation12}
\dot{V}_{L} \leq -s^{T}Qs -\frac{1}{4}W^{T}G_{\sigma}W+\frac{1}{2}W^{T}G_{\sigma}\tilde{W}_{a}+\frac{1}{4}G_{\epsilon}+\frac{1}{2}\tilde{W}_{a}^{T}\nabla_{s}\sigma G_{R}\nabla_{s}\epsilon^{T} 
	r\\-\tilde{W}_{c}^{T}\Gamma^{-1}\bigg(-k_{c1}\Gamma\frac{\omega}{\rho}(-\omega^{T}\tilde{W}_{c}+\frac{1}{4}\tilde{W}_{a}^{T}G_{\sigma}\tilde{W}_{a}-W^{T} \nabla_{s} \sigma y\tilde{\theta}+\frac{1}{2}W^{T}\nabla_{s}\sigma G_{R}\nabla_{s}\epsilon^{T}+\frac{1}{4}G_{\epsilon}-\nabla_{s}\epsilon F)\bigg) \\
	+\tilde{W}_{c}^{T}\Gamma^{-1}\bigg(\frac{1}{N}\Gamma\sum_{k=1}^{N}\frac{k_{c2}\omega_{k}}{\rho_{k}}\big(-\omega_{k}^{T}\tilde{W}_{c}+\frac{1}{4}\tilde{W}_{a}^{T}G_{\sigma k}\tilde{W}_{a}-(W^{T} \nabla_{s} \sigma_{k} y_{k}\tilde{\theta})+{\Delta}_{k}\big)\bigg) \\
	-\frac{\beta}{2}\tilde{W}_{c}^{T}\Gamma^{-1}\tilde{W}_{c}+\frac{1}{2}k_{c1}\tilde{W}_{c}^{T}\frac{\omega\omega^{T}}{\rho^{2}}\tilde{W}_{c}\\+\frac{1}{2}k_{c2}\tilde{W}_{c}^{T}\frac{1}{N}\sum_{k=1}^{N}\frac{\omega_{k}\omega_{k}^{T}}{\rho_{k}^{2}}\tilde{W}_{c}+k_{a1}\tilde{W}_{a}^{T}\tilde{W}_{c}-(k_{a1}+k_{a2})\tilde{W}_{a}^{T}\tilde{W}_{a} 
	+ k_{a2}\tilde{W}_{a}^{T}W\\-\tilde{W}_{a}^{T}\Bigg(\bigg(\frac{k_{c1}\omega}{4\rho}\hat{W}_{a}^{T}G_{\sigma}+\sum_{k=1}^{N}\frac{k_{c2}\omega_{k}}{4N\rho_{k}}\hat{W}_{a}^{T}G_{\sigma k}\bigg)^{T}\hat{W}_{c}\Bigg)-\lambda_{\min}\{Y_{f}\} \|\tilde{\theta}\|^{2}.
\end{multline}

Using Young's inequality, Cauchy-Schwarz inequality, and completion of squares, \eqref{lyapunovfunctionderivation12} can be bounded as 
\begin{multline}
    \dot{V}_{L}\leq-s^{T}Qs -k_{c2}\underline{c}\left(\varpi_{1}+\varpi_{2}+\varpi_{3}\right)\left\Vert \tilde{W}_{c}\right\Vert ^{2} -\left(k_{a1}+k_{a2}\right)\left(\varpi_{4}+\varpi_{5}+\varpi_{6}+\varpi_{7}\right)\left\Vert \tilde{W}_{a}\right\Vert ^{2}\\
    + \left(\frac{1}{2}\overline{W}\overline{\left\Vert G_{\sigma}\right\Vert }+\frac{1}{2}\left\Vert \nabla_{s}\epsilon G^{T}\nabla_{s}\sigma^{T}\right\Vert +k_{a2}\overline{W}+\frac{1}{4}\left(k_{c1}+k_{c2}\right)\overline{W}^{2}\overline{\left\Vert G_{\sigma}\right\Vert }\right)\left\Vert \tilde{W}_{a}\right\Vert\\
    +\left\Vert \tilde{W}_{c}\right\Vert \left(\left(k_{c1}+k_{c2}\right)\left\Vert \hat{\delta}\right\Vert \right)
    +\left(k_{a1}+\frac{1}{4}\left(k_{c1}+k_{c2}\right)\overline{W}\overline{\left\Vert G_{\sigma}\right\Vert }\right)
    \left\Vert \tilde{W}_{a}\right\Vert \left\Vert \tilde{W}_{c}\right\Vert
    +\frac{1}{4}\overline{\left \Vert G_{\epsilon} \right\Vert }\\
    +\frac{1}{4}\left(k_{c1}+k_{c2}\right)\overline{W}\overline{
    \left\Vert G_{\sigma}\right\Vert }\left\Vert \tilde{W}_{a}\right\Vert ^{2} 
    -\lambda_{\min}\{Y_{f}\} \|\tilde{\theta}\|^{2}  + \left(k_{5} r \right)\left( \frac{\|\tilde{\theta}\|^{2}}{2\epsilon} + \frac{\epsilon \|\tilde{W}_{c}\|^{2}}{2}\right).
\end{multline}
Provided the gains are selected based on the sufficient conditions in \eqref{eq:gain1}, \eqref{eq:gain2}, \eqref{eq:gain3} and \eqref{eq:gain4}, the orbital derivative can be upper-bounded as
\begin{equation}\label{eq:V_L}
    \dot{V}_{L}	\leq-v_{l}\left(\left\Vert Z\right\Vert \right),\quad\forall\left\Vert Z\right\Vert \geq v_{l}^{-1}\left(\iota\right),
\end{equation}
for all $t\geq T$ and $\forall Z\in B_{r}$. Using \eqref{eq:v_l}, \eqref{eq:gain4}, and \eqref{eq:V_L}, Theorem 4.18 in \cite{SCC.Khalil2002} can then be invoked to conclude that $Z$ is locally uniformly ultimately bounded in the sense that all trajectories starting from initial conditions bounded by $\left\Vert Z(T) \right\Vert \leq  \overline{v_{l}}^{-1}(\underline{v_{l}}(r))$, satisfy $\lim\sup_{t\to\infty}\left\Vert Z\left(t\right)\right\Vert \leq\underline{v_{l}}^{-1}\left(\overline{v_{l}}\left(v_{l}^{-1}\left(\iota\right)\right)\right).$ Furthermore, the concatenated state trajectories are bounded such that $\left\Vert Z\left(t\right)\right\Vert \in B_{r}$ for all $t\in\mathbb{R}_{\geq T}$. Since the estimates $\hat{W}_{a}$ approximate the ideal weights $W$, the policy $\hat{u}$ approximates the optimal policy $u^{*}$.
\end{proof}
\end{appendix}
\end{document}